\documentclass[a4paper,12pt]{amsart}
\usepackage{amsfonts}
\usepackage{amsthm}
\usepackage{amssymb}
\usepackage{amstext}
\usepackage{amsmath}
\usepackage{wasysym}

\usepackage{graphics}
\usepackage{epsfig}
\usepackage{amscd}
\usepackage[all]{xy}
\usepackage{latexsym}
\usepackage{graphicx}
\usepackage{multirow}
\usepackage{geometry}

\usepackage[utf8]{inputenc}

\newtheorem{theorem}{Theorem}[section]
\newtheorem{proposition}[theorem]{Proposition}
\newtheorem{corollary}[theorem]{Corollary}
\newtheorem{lemma}[theorem]{Lemma}

\theoremstyle{definition}
\newtheorem{definition}[theorem]{Definition}

\theoremstyle{remark}
\newtheorem{remark}[theorem]{Remark}

\newcommand{\la}{\left\langle}
\newcommand{\ra}{\right\rangle}
\newcommand{\lb}{\left (}
\newcommand{\rb}{\right )}
\newcommand{\A}{\mathcal{A}}
\newcommand{\E}{\mathcal{E}}

\newcommand{\W}{\mathcal{W}}

\renewcommand{\P}{\mathbb{P}^1}

\newcommand{\Z}{\mathbb{Z}}

\renewcommand{\la}{\left\langle}
\renewcommand{\ra}{\right\rangle}

\DeclareMathOperator{\Id}{Id}
\DeclareMathOperator{\Res}{Res}

\newcommand{\lv}{\left\langle\left |}
\newcommand{\rv}{\right |\right\rangle}

\newcommand{\cor}[1]{\left\langle #1 \right\rangle}

\newcommand{\Mbar}{{\overline{\mathcal{M}}}}

\newcommand{\bC}{{\mathbb{C}}}

\newcommand{\bQ}{{\mathbb{Q}}}

\newcommand{\cA}{{\mathcal{A}}}

\newcommand{\cD}{{\mathcal{D}}}

\newcommand{\cS}{{\mathcal{S}}}

\newcommand{\rar}{\rightarrow}
\newcommand{\lrar}{\longrightarrow}
\newcommand{\llar}{\longleftarrow}
\newcommand{\half}{{\frac{1}{2}}}

\numberwithin{equation}{section}

\begin{document}

\title[Quantum spectral curve for 
complex projective line]{Quantum spectral curve for the Gromov-Witten theory of the complex
projective line}

\author[Dunin-Barkowski]{P.~Dunin-Barkowski}

\address{P.~D.-B.: Korteweg-de~Vries Institute for Mathematics, University of Amsterdam, P.~O.~Box 94248, 1090 GE Amsterdam, The Netherlands and ITEP, Moscow, Russia}

\email{P.Dunin-Barkovskiy@uva.nl}

\thanks{\hspace{-.4cm}P.D.-B.\ is supported by a free competition grant of the NWO; M.M.\ is supported by 
NSF grants DMS-1104734 and
DMS-1309298; P.N.\ is supported by ARC grant DP1094328; A.P.\ and S.S.\ are supported by a Vici grant 
of the NWO; and S.S.\ is  supported by 
a Vidi grant of the NWO}

\author[Mulase]{M.~Mulase}

\address{M.M.: Department of Mathematics, University of California, Davis, CA 95616-8633, U.S.A.}

\email{mulase@math.ucdavis.edu}

\author[Norbury]{P.~Norbury}

\address{P.N.: Department of Mathematics and Statistics, The University of Melbourne, Victoria 3010, Australia}

\email{pnorbury@ms.unimelb.edu.au}

\author[Popolitov]{A.~Popolitov}

\address{A.P.: Korteweg-de Vries Institute for Mathematics, University of Amsterdam, Postbus 94248, 1090 GE Amsterdam, The Netherlands and ITEP, Moscow, Russia}

\email{A.Popolitov@uva.nl}

\author[Shadrin]{S.~Shadrin}

\address{S.S.: Korteweg-de Vries Institute for Mathematics, University of Amsterdam, Postbus 94248, 1090 GE Amsterdam, The Netherlands}

\email{S.Shadrin@uva.nl}

\begin{abstract}
We construct the quantum  curve  for the 
Gromov-Witten theory of the complex
projective line.
\end{abstract}

\maketitle

\tableofcontents

\section{Introduction}

The purpose of this paper is to construct
 the \emph{quantum curve} for 
the Gromov-Witten invariants of the 
complex projective line $\P$.
Quantum curves are conceived in the physics
literature, including 
\cite{ADKMV06, DHS, DHSV, 
 EM, GS12, Hollands}. They quantize the
 \emph{spectral curves} of the theory, and 
 are conjectured to capture the
information of many topological invariants, 
such as certain Gromov-Witten invariants,
quantum knot invariants, and cohomology of 
instanton moduli spaces
 for $4$-dimensional gauge theory.
In this paper we show that
the conjecture is indeed 
true for the Gromov-Witten theory of
$\P$.

\subsection{Spectral curves and quantum curves}

When spectral curves  appear  
in mathematics, they  take
various  different forms,
and even look as totally different objects. For example,
they can be the mirror curve
of a toric Calabi-Yau $3$-fold, the 
$SL_2$-character
variety of the fundamental group of a knot 
complement, or a Seiberg-Witten curve. In the
context of the Gromov-Witten theory of $\P$, 
it is the \emph{Landau-Ginzburg model}
\begin{equation}
\label{eq:LGM}
x = z + \frac{1}{z},
\end{equation}
which is the homological mirror dual of 
$\P$ with respect to
the standard K\"ahler structure.
Our main theorem (Theorem~\ref{thm:main}
below)
states that the quantization of \eqref{eq:LGM},
which we call a quantum curve,
characterizes the exponential generating 
function of Gromov-Witten invariants of
$\P$.

In a purely algebro-geometric setting, a quantum 
curve can be understood in the following 
way \cite{DM13}.
Let $C$ be a non-singular complex 
projective algebraic curve, and
$\eta$ the tautological $1$-form on the
cotangent bundle $T^*C$. A 
spectral curve
$\Sigma$ is a complex $1$-dimensional
subvariety 
\begin{equation}
\label{eq:spectral}
\begin{CD}
\iota :\Sigma @>>> T^*C
\\
&&@VV{\pi}V
\\
&&C
\end{CD}
\end{equation}
in the cotangent bundle, which is automatically
a Lagrangian subvariety with respect to the standard
symplectic form $-d\eta$. A quantum curve
is an $\hbar$-deformed
 $D$-module on the $1$-parameter formal
family $C[[\hbar]]$ of the curve $C$,
whose \emph{semi-classical limit} coincides with 
the spectral curve $\Sigma$. On an affine 
piece of the base curve $C$ with coordinate
$x$, we can choose a generator $P$ of the
$D$-module and consider a 
\emph{Schr\"odinger-like} equation
\begin{equation}
\label{eq:Sch-intro}
P(x,\hbar)\Psi(x,\hbar) = 0.
\end{equation}
The construction of the quantum curve
in this  setting
is established in \cite{DM13}
for $SL(2,\bC)$ Hitchin
fibrations.

The geometric situation we consider in this paper is slightly different. 
Instead of the cotangent bundle $T^*C$ in \eqref{eq:spectral}, we have a surface $X$ equipped with a 
$\bC^*$-invariant holomorphic symplectic form and a spectral curve $\Sigma$ is mapped into it.  The base curve $C$ is replaced by the quotient $X/\bC^*$.  For example, if a curve $C$ admits a $\bC^*$-action then the natural holomorphic symplectic form on $X=T^*C$ is $\bC^*$-invariant.  In local coordinates for $T^*C$, the $\bC^*$-action is given by $c\cdot(w,z)=(cw,c^{-1}z)$ and the symplectic form is given by
 $dx\wedge (dz/z)$ where $x=wz$ is the quotient map $(w,z)\mapsto wz$ by the $\bC^*$-action.
Reflecting the $\bC^*$-action, the
quantum curve \eqref{eq:Sch-intro} becomes
a differential equation of \emph{infinite} order, or a \emph{difference}
equation.

We present in this paper the first rigorous
example of
a direct connection between Gromov-Witten
theory and  quantum curves. Our construction 
requires the  
fermionic Fock space representation of
the Gromov-Witten invariants \cite{OP06}, and 
a subtle combinatorial analysis based on 
representation theory of symmetric groups.

\subsection{Main theorem}

Let $\Mbar_{g,n}(\P,d)$ denote the 
moduli space of stable maps of degree $d$
from an $n$-pointed genus $g$ curve to
$\P$. This is an algebraic stack of 
dimension $2g-2+n +2d$. The dimension
reflects the fact that a generic map
from an algebraic curve to $\P$ has 
only simple ramifications, which we
can see from the Riemann-Hurwitz 
formula. The  descendant 
Gromov-Witten invariants of 
$\P$ are defined by 
\begin{equation}
\label{eq:GW}
\la \prod_{i=1}^n
\tau_{b_i}(\alpha_i)\ra_{g.n} ^d
:=
\int _{[\Mbar_{g,n}(\P,d)]^{vir}}
\prod_{i=1}^n \psi_i ^{b_i} ev_i^*(\alpha_i),
\end{equation}
where 
$[\Mbar_{g,n}(\P,d)]^{vir}$ is the virtual 
fundamental class of the moduli space,
\begin{equation*}
ev_i:\Mbar_{g,n}(\P,d)\lrar \P
\end{equation*}
is a natural morphism defined by
evaluating a stable map at the $i$-th marked 
point of the source curve, $\alpha_i\in H^*(\P,\bQ)$
is a cohomology class of the target $\P$,
and $\psi_i$ is  the tautological cotangent 
class in $H^2(\Mbar_{g,n}(\P,d),\bQ)$.
We denote by $1$ the generator of
$H^0(\P,\bQ)$, and by $\omega\in
H^2(\P,\bQ)$ the Poincar\'e dual to the 
point class.
We assemble the Gromov-Witten invariants
into particular generating functions as follows.
For every $(g,n)$ in the stable sector
$2g-2+n>0$, we define the \emph{free energy} of 
type $(g,n)$
by
\begin{equation}
\label{eq:Fgn-intro}
F_{g,n}(x_1,\dots,x_n)
:= 
\la \prod_{i=1}^n \left(
-\frac{\tau_0(1)}{2} - \sum_{b=0}^\infty
\frac{b!\tau_b(\omega)}{x_i^{b+1}}
\right)
\ra_{g,n}.
\end{equation}
Here the degree $d$ is determined by the 
dimension condition of the 
cohomology classes to be integrated
over the virtual fundamental class.
We note that \eqref{eq:Fgn-intro} contains 
the class $\tau_0(1)$. 
For unstable geometries, we introduce two 
functions
\begin{align}
\label{eq:S0-intro}
S_0(x) &:=
x-x\log x +
\sum_{d=1}^\infty 
\la -\frac{(2d-2)!\tau_{2d-2}(\omega)}
{x^{2d-1}}
\ra_{0,1}^d,
\\
\label{eq:S1-intro}
S_1(x) &:=
-\half \log x+
\half \sum_{d=0}^\infty 
\la 
\left(
-\frac{\tau_0(1)}{2} - \sum_{b=0}^\infty
\frac{b!\tau_b(\omega)}{x^{b+1}}
\right)^2
\ra_{0,2}^d.
\end{align}
The appearance of the extra terms, in particular
the $\log x$ terms, will be explained in 
Section~\ref{sec:shift-of-variable}.
We shall prove the following.

\begin{theorem}[Main Theorem]
\label{thm:main}
The wave function
\begin{equation}
\label{eq:Psi}
\Psi(x,\hbar)
:=\exp\left(
\frac{1}{\hbar}S_0(x) + S_1(x)
+\sum_{2g-2+n>0}\frac{\hbar^{2g-2+n}}{n!}
F_{g,n}(x,\dots,x)
\right)
\end{equation}
satisfies the quantum curve equation
of an infinite order
\begin{equation}
\label{eq:QCE}
\left[
\exp\left(
\hbar\frac{d}{dx}
\right)
+
\exp\left(
-\hbar\frac{d}{dx}
\right)
-x
\right]
\Psi(x,\hbar) = 0.
\end{equation}
Moreover, 
the free energies $F_{g,n}(x_1,\dots,x_n)$
as functions in $n$-variables,
 and hence 
all the Gromov-Witten invariants \eqref{eq:GW},
can be recovered 
from the equation
\eqref{eq:QCE} alone,
using the mechanism of the
\textbf{topological recursion}
of \cite{CEO, EO07}.
\end{theorem}

\begin{remark}Put 
\begin{equation}
\label{eq:Sm-intro}
S_m(x) := \sum_{2g-2+n=m-1}\frac{1}{n!}
F_{g,n}(x,\dots,x).
\end{equation}
Then our wave function is of the form
\begin{equation}
\label{eq:WKB}
\Psi(x,\hbar) = \exp\left( \sum_{m=0}^\infty
\hbar^{m-1} S_m(x)\right),
\end{equation}
which provides the WKB approximation 
of the quantum curve equation 
\eqref{eq:QCE}. Thus the significance of 
\eqref{eq:Fgn-intro} is that the 
exponential generating function 
\eqref{eq:Psi} of the
descendant Gromov-Witten invariants of 
$\P$ gives the  solution to the \emph{exact}
WKB analysis for the difference equation
\eqref{eq:QCE}.
\end{remark}

\begin{remark}
For the case of Hitchin fibrations \cite{DM13},
the Schr\"odinger-like equation
\eqref{eq:Sch-intro} is a direct consequence 
of the generalized topological recursion. In our
current context, the topological recursion does 
not play any role in establishing 
\eqref{eq:QCE}. 
\end{remark}

\begin{remark}
Although the shape of the operator 
in \eqref{eq:QCE} has a similarity with the 
Lax operator of the
Toda lattice equations that control the 
Gromov-Witten invariants of 
$\P$ \cite{OP06}, we are unable to 
find any direct relations between these
two apparently different equations. 
We present a detailed
comparison of these equations in 
Section~\ref{sec:Toda}.
\end{remark}

\subsection{WKB approximation, 
topological recursion, and representation theory}

The  WKB analysis provides a 
perturbative quantization method
of a classical mechanical problem. We can
recover the classical problem corresponding 
to \eqref{eq:QCE} by taking its 
semi-classical limit, which is the 
singular perturbation limit
\begin{multline}
\label{eq:SCL}
\lim_{\hbar\rar 0}
\left(
e^{-\frac{1}{\hbar}S_0(x)}
\left[
\exp\left(
\hbar\frac{d}{dx}
\right)
+
\exp\left(
-\hbar\frac{d}{dx}
\right)
-x
\right]
e^{\frac{1}{\hbar}S_0(x)}
e^{ \sum_{m=1}^\infty
\hbar^{m-1} S_m(x)}
\right)
\\
=
\left(e^{S_0'(x)}+e^{-S_0'(x)}-x
\right)
e^{S_1(x)}=0.
\end{multline}
In terms of new variables
$y(x) = S_0'(x)$ and $z(x) = e^{y(x)}$, 
the semi-classical limit  gives us 
an equation for the spectral curve
$$
z\in \Sigma = \bC^*\subset \bC \times \bC^*
\overset{\exp}{\llar} T^*\bC = \bC^2
\owns (x,y)
$$ by
\begin{equation}
\label{eq:spectral-intro}
\begin{cases}
x = z+\frac{1}{z}\\
z = e^y
\end{cases}.
\end{equation}
This is the reason we consider \eqref{eq:QCE}
as the quantization of the Laudau-Ginzburg
model \eqref{eq:LGM}.

It was conjectured in \cite{NS11} that the
stationary Gromov-Witten theory of $\P$ should
satisfy the topological recursion
of \cite{CEO,EO07}
with respect to the spectral curve
\eqref{eq:spectral-intro}. We refer to 
\cite{DM13,DMSS12, NS11} for a mathematical 
formulation of the topological recursion.
The conjecture
is solved in \cite{DOSS12} as a corollary to 
its main theorem, which 
establishes the correspondence between the 
topological 
recursion and the Givental formalism.

The quantum curve equation
\eqref{eq:QCE}
determines only the function $\Psi$, and by
the $\hbar$-expansion, each coefficient
$S_m(x)$. But then how do we possibly 
recover $F_{g,n}$ for each $(g,n)$ 
as a function in 
$n$ variables? Here comes the 
significance of the topological recursion 
of \cite{CEO,EO07}, which was established
in \cite{DOSS12} for the case of the Gromov-Witten
theory of $\P$. 
The scenario goes as follows. First we 
note that the semi-classical limit of \eqref{eq:QCE}
identifies the spectral curve \eqref{eq:spectral-intro}.
We then launch the topological recursion formalism
of \cite{CEO,EO07} for this particular
spectral curve,
and obtain symmetric differential $n$-forms
$\W_{g,n}(z_1,\dots,z_n)$ on $\Sigma^n$.
In this paper we will present a canonical way
to \emph{integrate} these $n$-forms, 
which yields the free energy 
$F_{g,n}(x_1,\dots,x_n)$ for every $(g,n)$
subject to $2g-2+n>0$. In this sense
the single equation \eqref{eq:QCE}
knows the information of all Gromov-Witten
invariants \eqref{eq:GW}. This shows the
power of quantum curves.

The key discovery of the present paper is that
the quantum curve equation \eqref{eq:QCE}
is equivalent to a  recursion equation
\begin{equation}
\label{eq:X-intro}
\frac{x}{\hbar}
\left(
e^{-\hbar\frac{d}{dx}} -1
\right) X_d(x,\hbar)
+\frac{1}{1+\frac{x}{\hbar}}
e^{\hbar\frac{d}{dx}}X_{d-1}(x,\hbar)=0
\end{equation}
for a rational function
\begin{equation}
\label{eq:Xd-intro}
X_d(x,\hbar) = \sum_{\lambda\vdash d}
\left(
\frac{\dim \lambda}{d!}
\right)^2
\prod_{i=1}^{\ell(\lambda)}
\frac{x+(i-\lambda_i)\hbar}{x+i\hbar}.
\end{equation}
Here $\lambda$ is a partition of $d\ge 0$
with parts $\lambda_i$ and 
$\dim \lambda$ denotes the dimension 
of the irreducible representation of 
the symmetric group $S_d$ 
characterized  by $\lambda$.

\subsection{Organization of the paper}

This paper is organized as follows.
In Section~\ref{sec:Fgn} we start with 
a solution $\W_{g,n}$ to the topological 
recursion equation
with respect to the spectral curve
$\Sigma$ of
\eqref{eq:spectral-intro}. It is a symmetric
differential form of degree $n$ on
$\Sigma^n$. We then propose a
unique mechanism to integrate $\W_{g,n}$
into a rational function. 
The  goal of this section is to show that 
this primitive function is identical to 
\eqref{eq:Fgn-intro}.
Then in
Section~\ref{sec:shift-of-variable},
we re-write $\Psi(x,\hbar)$ in a different manner,
only involving stationary Gromov-Witten
invariants of $\P$. 
This formula  allows us to express it in terms
of a \emph{semi-infinite wedge product} in
Section~\ref{sec:reduction-to-semi-infinite-wedge}.
Using this formalism, we  reduce the
quantum curve equation 
\eqref{eq:QCE} to a combinatorial equation
\eqref{eq:X-intro} in
Section~\ref{sec:combinatorial}.  
Equation~\eqref{eq:X-intro} is then proved in
 Section~\ref{sec:key} using representation 
 theory of $S_d$, which in tern  establishes 
 \eqref{eq:QCE}. For completeness, 
 we give an expression of \eqref{eq:Xd-intro}
 in terms of  special values of  the Laguerre 
 polynomials in  Section~\ref{sec:Laguerre}.
Section~\ref{sec:Toda} is devoted to 
the comparison of 
\eqref{eq:QCE} and the Toda lattice equations of
\cite{OP06}, in terms of the functions
$X_d$ of \eqref{eq:Xd-intro}.

\section{The functions $F_{g,n}$ in terms of Gromov-Witten invariants}
\label{sec:Fgn}

The significance of the idea of quantum curves
is that the single equation \eqref{eq:Sch-intro}
captures all information of the topological 
invariants of the theory. The key process
from this single equation to the topological
invariants is the 
\emph{integral form}
of the mechanism known as
the topological recursion of \cite{CEO,EO07}. 
We refer to \cite{DM13,DMSS12,NS11} for
mathematical formulation of the topological
recursion. This section is devoted to 
providing the unique mechanism to 
integrate the topological recursion, for 
the context of the Gromov-Witten theory 
of $\P$.

Let us begin with a solution 
$\W_{g,n}(z_1,\dots,z_n)$
to the topological recursion  of
\cite{CEO,DOSS12,EO07} associated with 
the spectral curve $\Sigma = \bC^*$ defined by
\begin{equation}
\label{eq:xy}
\begin{cases}
x(z)  = z+\frac{1}{z} \\
y(z)  = \log z
\end{cases}.
\end{equation}
This means that 
 symmetric differential forms
 $\W_{g,n}(z_1,\dots,z_n)$ of degree $n$ on
$\Sigma^n$
for $(g,n)$ in the stable range $2g-2+n>0$
are inductively defined by the 
following recursion formula:
\begin{multline}
\label{eq:CEO}
\W_{g,n}(z_1,\dots,z_n)\\
=
\frac{1}{2\pi i} \oint_{z=\pm 1}
\frac{\int_z ^{1/z} \W_{0,2}(\;\cdot\;,z_1)}
{\W_{0,1}(1/z) - \W_{0,1}(z)}
\Bigg[\W_{g-1,n+1}(z,1/z,z_2,\dots,z_n)
\\
+
\sum_{\substack{g_1+g_1=g\\
I\sqcup J=\{2,\dots,n\}}} ^{\text{stable}}
\W_{g_1,|I|+1}(z,z_I)\W_{g_2,|J|+1}
(1/z,z_J)
\Bigg],
\end{multline}
where the residue integral is taken with respect to
the variable $z\in \Sigma$ on 
two small, positively oriented,
closed loops around $z=1$ and $z=-1$,
and for the index set $I\subset \{2,\dots,n\}$,
we denote by
$|I|$  its cardinality, 
and $z_I = (z_i)_{i\in I}$. 
For $(g,n)$ in the unstable range, we define
\begin{align}
\label{eq:W01}
&\W_{0,1}(z) :=y(z) dx(z),\\
\label{eq:W02}
&\W_{0,2}(z_1,z_2) := \frac{dz_1 dz_2}
{(z_1-z_2)^2} -\frac{dx(z_1) dx(z_2)}
{(x(z_1)-x(z_2))^2}.
\end{align}
The goal of this section is to 
derive the  integral $F_{g,n}(z_1,\dots,z_n)$
of $\W_{g,n}(z_1,\dots,z_n)$ in a consistent and
unique way that has the 
$x$-variable expansion \eqref{eq:Fgn-intro}.

\begin{remark}
The second term of the right-hand side of
\eqref{eq:W02} does not
play any role in the topological recursion
\eqref{eq:CEO}. It is included here for 
the consistency of the primitive $F_{0,2}(z_1,z_2)$ 
to be discussed in Section~\ref{sec:shift-of-variable}.
\end{remark}

\begin{definition}
\label{def:Fgn}
For $2g-2+n>0$, we define the
 \emph{primitive} $F_{g,n}(z_1,\dots,z_n)$ 
 of the $n$-form $\W_{g,n}(z_1,\dots,z_n)$ to be
a rational  function on $\Sigma^n$
that satisfies the following conditions:
\begin{align}
\label{eq:primitive}
&d_1\cdots d_n F_{g,n}(z_1,\dots,z_n) 
=\W_{g,n}(z_1,\dots,z_n);\\
\label{eq:skew}
& F_{g,n}(z_1,\dots,z_{i-1},1/z_i,
z_{i+1},\dots,z_n) = -F_{g,n}(z_1,\dots,z_n),
\quad i=1,\dots,n;\\
\label{eq:initial}
&F_{g,n}(z_1,\dots,z_n)\big|_{z_1=\cdots=z_n=0}
=0.
\end{align}
If it exists, then it is unique.
\end{definition}

From now on, we need to relate functions or 
differential forms defined 
on the spectral curve $\Sigma= \bC^*$
of \eqref{eq:xy}
and on the base curve $\bC$. We recall \cite{DMSS12}
that the inverse function of \eqref{eq:LGM}
for the branch near $z=0$ and  $x=\infty$ is given by
the generating function of the Catalan numbers
\begin{equation}
\label{eq:Catalan}
z=z(x) = \sum_{m=0}^\infty \frac{1}{m+1}
\binom{2m}{m} \frac{1}{x^{2m+1}}.
\end{equation}
By abuse of notation, for a function or a differential
form $f(z)$ on $\Sigma$, we denote the pull-back via
\eqref{eq:Catalan}  simply by $f(x) := f(z(x))$.

It is established in \cite{DOSS12, NS11}
that the solution $\W_{g,n}$ of the topological 
recursion has the following $x$-variable expansion
 in terms of the stationary 
 Gromov-Witten invariants of $\P$:
\begin{equation}
\label{eq:NS}
\W_{g,n}(x_1,\dots,x_n)=
\la \prod_{i=1}^n \left(\sum_{b=0}^\infty (b+1)!\, \tau_{b} (\omega)\,\frac{dx_i}{x_i^{b+2}} \right) \ra_{g,n} .
\end{equation}
There is no systematic mechanism to integrate 
this expression to obtain \eqref{eq:Fgn-intro}.
Instead, we establish the following
theorem in this section.

\begin{theorem} 
\label{thm:Fgn} 
For every $(g,n)$ in the
stable sector $2g-2+n>0$, 
there exists a primitive $F_{g,n}(z_1,\dots,z_n)$
in the sense
of Definition~\ref{def:Fgn}, such that
its $x$-variable expansion is given by
\begin{equation}
\label{eq:Fgn}
F_{g,n}(x_1,\dots,x_n)=
\la \prod_{i=1}^n \left(-\frac{\tau_0(1)}{2}-\sum_{b=0}^\infty\frac{b! \tau_{b} (\omega) }{x_i^{b+1}} \right) \ra_{g,n} .
\end{equation}
\end{theorem}

\begin{remark}
We need a different treatment for the unstable
primitives $F_{0,1}(z)$ and $F_{0,2}(z_1,z_2)$.
They are calculated in 
Section~\ref{sec:shift-of-variable}.
\end{remark}

The rest of this section is devoted to proving
this theorem. We start with recalling some results 
of \cite{DOSS12}. 
The most important one is  the formula for 
$\W_{g,n}(z_1,\dots,z_n)$ in terms of the
auxiliary  functions $W^i_d(z)$ (defined below) 
with the \emph{ancestor} Gromov-Witten invariants
  as its coefficients.  We will then
  prove the existence of  the anti-symmetric primitives
   of the functions $W^i_d$, and their $x$-expansions. 
   This will then lead us to the proof of
    the above theorem, where we will 
    also utilize  the known relations between the 
    ancestor and the
    descendant Gromov-Witten invariants.

\subsection{Some results from \cite{DOSS12}}

The ancestor Gromov-Witten invariants 
of $\P$ we need
are 
\begin{equation}
\label{eq:ancestor}
\left\langle \prod_{i=1}^n\bar{\tau}_{b_i}(\alpha_i)
\right\rangle^{d}_{g,n}
:=\int_{[\Mbar_{g,n}(\P,d)]^{vir}}\prod_{i=1}^n \bar{\psi}_i^{b_i}ev_i^\ast(\alpha_i),
\end{equation}
where $\bar{\psi}_i$ denotes the pull back of 
the cotangent class on $\Mbar_{g,n}$ by the
natural forgetful morphism
$$
\Mbar_{g,n}(\P,d)\lrar \Mbar_{g,n}.
$$
Since we adopt a quantum field theoretic point of
view in calculating Gromov-Witten invariants,
we often call them \emph{correlators} in this
paper. The ancestor and descendant correlators 
do not agree. We will give a formula to determine
 one from the other in \eqref{eq:anc_desc_rel}.

Let us  define
\begin{align}
W_0^1(z)  &:= \dfrac{dz}{(1-z)^2},\\
W_0^2(z)  &:= \dfrac{i dz}{(1+z)^2},
\\
\label{eq:Wkdef}
W^i_k (z) &:= d\left( \left(-2\frac{d}{d x(z)}\right)^k \int W_0^i (z)\right),\qquad i=1,2; \quad k\ge 0.
\end{align}
Then for $g\geq 0$ and $n\geq 1$
with $2g-2+n>0$, from 
Theorem 4.1 of \cite{DOSS12}  
(as shown in the proof of Theorem 5.2 
of \cite{DOSS12}), we have
\begin{equation}
\label{eq:Wgn_anc}
\W_{g,n}(z_1,\dots,z_n)=\sum_{\vec{d},\;\vec{i}}\cor{\bar{\tau}_{d_1}(\tilde{e}_{i_1})\dots\bar{\tau}_{d_n}(\tilde{e}_{i_n})}_g \dfrac{W_{d_1}^{i_1}(z_1)}{2^{d_1}\sqrt{2}}\dots \dfrac{W_{d_n}^{i_n}(z_n)}{2^{d_n}\sqrt{2}}.
\end{equation}
Here  the sum over $\vec{d}$ and $\vec{i}$ 
are taken over all integer values $0\leq d_k$ and 
 $i_k =1, 2$. Note that
 the coefficients of this expansion
  are the \emph{ancestor} Gromov-Witten invariants. 
  The cohomology basis for $H^1(\P,\bQ)$
  is normalized as follows. First we denote
  by $e_1=1$ and $e_2 = \omega$. 
  Using the normalization matrix 
  \begin{equation*}
A=\dfrac{1}{\sqrt{2}}\,\begin{pmatrix}
1 & -i \\
1 & i
\end{pmatrix},
\end{equation*}
we define
\begin{equation*}
\tilde{e}_i = \left(A^{-1}\right)^\mu_i e_\mu.
\end{equation*}
In this section we
use the Einstein convention and take
 summation over repeated indices.

With the help of the Givental formula, 
Proposition 5.1 of \cite{DOSS12} relates 
 the ancestor and the descendant correlators for 
 $\P$ by
 \begin{equation}
\begin{aligned}
\label{eq:anc_desc_rel}
&\sum_{\vec{d},\;\vec{i}}\cor{\bar{\tau}_{d_1}(\tilde{e}_{i_1})\dots\bar{\tau}_{d_n}(\tilde{e}_{i_n})}_g v^{d_1,i_1}\dots v^{d_n,i_n} \\
= &\sum_{\vec{d},\;\vec{\mu}}\cor{\tau_{d_1}(e_{\mu_1})\dots\tau_{d_n}(e_{\mu_n})}_g t^{d_1,\mu_1}\dots t^{d_n,\mu_n}, 
\end{aligned}
\end{equation}
where $v^{d,i}$ and $t^{d,\mu}$ are formal variables related by the following formula:
\begin{equation}
\label{eq:vtrel}
v^{d,i} = A^i_{\mu}\sum_{m=d}^{\infty}(\cS_{m-d})^{\mu}_{\nu} t^{m,\nu}.
\end{equation}
Here $(\cS_k)^\mu_\nu$ are the matrix elements 
of the Givental $S$-matrix
and defined by
\begin{equation}
\begin{aligned}
\label{eq:S-matrix}
\cS(\zeta^{-1}) = \sum_{k=0}^\infty
\cS_k \zeta^{-k}
&= \Id + \zeta^{-1} \cdot
\left(\begin{matrix}
0 & 0 \\
1 & 0
\end{matrix}\right)
\\ 
&+\sum_{k=1}^\infty
\frac{\zeta^{-2k}}{(k!)^2}
\left(\begin{matrix}
1-2k\left(\frac{1}{1}+\cdots +\frac{1}{k}\right) & 0 \\
0 & 1
\end{matrix}\right)
\\ 
& +\sum_{k=1}^\infty
\frac{\zeta^{-2k-1}}{(k!)^2}
\left(\begin{matrix}
0 & -2\left(\frac{1}{1}+\cdots+ \frac{1}{k}\right) \\
\frac{1}{k+1} & 0
\end{matrix}\right).
\end{aligned}
\end{equation}
In the proof of Theorem 5.2 of \cite{DOSS12} it was 
shown that the $x^{-1}$-expansion of $W^{i}_d(z)$ 
near $z=0$ was given by the following formula:
\begin{equation}
W^{i}_d(z)= 2^d \sqrt{2}\,A^i_{\mu}\sum_{m=d}^{\infty}(\cS_{m-d})^{\mu}_{\nu}\ \delta^{\nu}_{2}\,  (m+1)!\, \dfrac{dx}{x^{m+2}},
\end{equation}
where $\delta^i_j$ is the Kronecker delta symbol.
The above formula, together with formulas \eqref{eq:Wgn_anc}-\eqref{eq:S-matrix}, implies  \eqref{eq:NS}.

The first step of integrating $\W_{g,n}$ is 
to identify a suitable primitive of 
the differential $1$-forms $W_d^i(z)$.

\begin{proposition}
\label{prop:theta}
For given $i=1,2$ and $d\ge 0$, 
there exists a uniquely defined rational
function 
$\theta^i_d(z)$ on $\Sigma$ such that
\begin{align}
\label{eq:theta1cond}
d\theta^{i}_d(z) &= W^{i}_d(z),\\
\theta^i_d(1/z) & =-\,\theta^i_d(z).
\label{eq:theta2cond}
\end{align}
Moreover, the $x^{-1}$-expansion of $\theta^i_d(z)$ near $z=0$ is given by 
\begin{equation}
\label{eq:thetaexp}
\theta^{i}_d (z(x))= 2^d \sqrt{2}\, A^i_{\mu}\sum_{m=d}^{\infty}(\cS_{m-d})^{\mu}_{\nu}\ \left( - \delta^{\nu}_{1}\delta^{m}_{0}\, \dfrac{1}{2} - \delta^{\nu}_{2}\,  m!\, \dfrac{1}{x^{m+1}} \right).
\end{equation}
\end{proposition}

\subsection{Proof of Proposition~\ref{prop:theta}}
It is easy to see by direct computation that the 
rational functions
\begin{equation}
\begin{aligned}
\label{eq:theta0def}
\theta_0^1 &:= \frac{1}{1-z}-\frac{1}{2}\\ 
\theta_0^2 &:= -\frac{i}{1+z}+\frac{i}{2}
\end{aligned}
\end{equation}
are the unique solutions of \eqref{eq:theta1cond}
and \eqref{eq:theta2cond} for $d=0$.

Equation \eqref{eq:Wkdef}, together with condition \eqref{eq:theta1cond}, implies that if $\theta^i_d(z)$ exists, then it has to satisfy
\begin{equation}
\label{eq:thetakdef}
\theta^i_d (z) = \left(-2\frac{d}{d x(z)}\right)^d \theta_0^i(z).
\end{equation}
Since $x$ is symmetric under 
the coordinate change $z\longmapsto 1/z$, 
we see that the right-hand side of equation 
\eqref{eq:thetakdef} satisfies  
\eqref{eq:theta2cond}. 
This means that $\theta_d^i(z)$ \emph{defined} by
 \eqref{eq:thetakdef} is, for given $i$ and $d$,
 indeed the unique solution of
  \eqref{eq:theta1cond} and \eqref{eq:theta2cond}.

We denote by $\tilde{\theta}^i_d$
the right-hand side of  \eqref{eq:thetaexp}. 
We wish  to prove that the 
$x^{-1}$-expansion of $\theta^i_d(z)$ 
near $z=0$  is given by $\tilde{\theta}^i_d$. 
Let us introduce the following notation:
\begin{equation}
\eta^\mu_d := \dfrac{1}{2^{d}\,\sqrt{2}}\, \left(A^{-1}\right)^{\mu}_i\, \theta^i_d.
\end{equation}
Then we have
\begin{align}
\label{eq:eta0}
\eta_0 &= \left( \frac{1}{1-z^2}-\frac{1}{2}, \frac{z}{1-z^2} \right),\\ \label{eq:etak}
\eta^\mu_k (z) &= \left(-\frac{d}{d x(z)}\right)^k \eta_0^\mu,
\end{align}
and condition \eqref{eq:thetaexp} becomes equivalent to the condition that the $x^{-1}$-expansion of $\eta^\mu_d$ near $z=0$ is equal to $ \tilde{\eta}^\mu_d$, where
\begin{equation}
\label{eq:etaexp}
\tilde{\eta}^\mu_d := \sum_{m=d}^{\infty}(\cS_{m-d})^{\mu}_{\nu}\ \left( - \delta^{\nu}_{1}\delta^{m}_{0}\, \dfrac{1}{2} - \delta^{\nu}_{2}\,  m!\, \dfrac{1}{x^{m+1}} \right).
\end{equation}
Let us prove formula \eqref{eq:etaexp} for $d=0$. Note that $\cS_0 = \Id$, so for the constant term of $\tilde{\eta}_0$ we have 
\begin{equation}
\left[\dfrac{1}{x^0}\right]\tilde{\eta}^\mu_0 = - \delta^{\mu}_{1} \dfrac{1}{2}.
\end{equation}
It is easy to see from \eqref{eq:eta0} that $\eta_0^i$ has the same constant term at $z=0$.

For $k\geq 1$ we have
\begin{equation}
\begin{aligned}
\label{eq:tetacoefs}
\left[\dfrac{1}{x^{2k-1}}\right]\tilde{\eta}^1_{0} &= -(2k-2)!\,(\cS_{2k-2})^1_2 = 0,\\ 
\left[\dfrac{1}{x^{2k-1}}\right]\tilde{\eta}^2_{0} &=- (2k-2)!\,(\cS_{2k-2})^2_2 = -\dfrac{(2k-2)!}{((k-1)!)^2},\\ 
\left[\dfrac{1}{x^{2k}}\right]\tilde{\eta}^1_{0} &=-(2k-1)!\,(\cS_{2k-1})^1_2 =- \dfrac{(2k-1)!}{k!\, (k-1)!},\\ 
\left[\dfrac{1}{x^{2k}}\right]\tilde{\eta}^2_{0} &= -(2k-1)!\,(\cS_{2k-1})^2_2 = 0 .
\end{aligned}
\end{equation}
For the corresponding coefficients in the 
$x^{-1}$-expansion of $\eta_0^{\mu}$ near $z=0$ we have ($k\geq 1$):
\begin{equation}
\begin{aligned}
\label{eq:etacoefs}
\mathop{\Res}_{z=0}x^{2k-2}(z)\,\eta^1_{0}\, dx(z) &= -\mathop{\Res}_{z=0}z^{-2k}\left(1+z^2\right)^{2k-2}dz = 0,\\ 
\mathop{\Res}_{z=0}x^{2k-2}(z)\,\eta^2_{0}\, dx(z) &=- \mathop{\Res}_{z=0}z^{-2k+1}\left(1+z^2\right)^{2k-2}dz = -\dfrac{(2k-2)!}{((k-1)!)^2},\\ 
\mathop{\Res}_{z=0}x^{2k-1}(z)\,\eta^1_{0}\, dx(z) &=- \mathop{\Res}_{z=0}z^{-2k-1}\left(1+z^2\right)^{2k-1}dz = -\dfrac{(2k-1)!}{k!\, (k-1)!},\\ 
\mathop{\Res}_{z=0}x^{2k-1}(z)\,\eta^2_{0}\, dx(z) &= -\mathop{\Res}_{z=0}z^{-2k}\left(1+z^2\right)^{2k-1}dz = 0.
\end{aligned}
\end{equation}
We see that the coefficients in \eqref{eq:tetacoefs} precisely coincide with the ones in \eqref{eq:etacoefs}. This implies that the 
$x^{-1}$-expansion of $\eta^\mu_0$ is indeed given by $\tilde{\eta}^{\mu}_0$.

By virtue of \eqref{eq:etak}, we see
 that the $x^{-1}$-expansion of $\eta_k^\mu$ near $z=0$ is given by the following formula (for $k\geq 1$):
\begin{multline*}
 \left(-\frac{d}{d x}\right)^k \eta_0^\mu 
=\sum_{m=0}^{\infty}(\cS_{m})^{\mu}_{\nu}\ \left( - \delta^{\nu}_{2}\,  (m+k)!\, \dfrac{1}{x^{m+k}} \right) 
\\
=\sum_{m=d}^{\infty}(\cS_{m-k})^{\mu}_{\nu}\ \left( - \delta^{\nu}_{2}\,  m!\, \dfrac{1}{x^{m+1}} \right).
\end{multline*}
This coincides with the formula for $\tilde{\eta}_k^\mu$ for $k\geq 1$. Thus, we have proved that the $x^{-1}$-expansion of $\eta^\mu_k$ is given by $\tilde{\eta}^{\mu}_k$, which, in turn, implies that Equation \eqref{eq:thetaexp} holds.
This concludes the proof of the proposition.

\subsection{Proof of Theorem~\ref{thm:Fgn}}
Recall Equation \eqref{eq:Wgn_anc} for $\W_{g,n}$:
\begin{equation*}
\W_{g,n}(z_1,\dots,z_n)=\sum_{\vec{d},\;\vec{i}}\cor{\bar{\tau}_{d_1}(\tilde{e}_{i_1})\dots\bar{\tau}_{d_n}(\tilde{e}_{i_n})}_g \dfrac{W_{d_1}^{i_1}(z_1)}{2^{d_1}\sqrt{2}}\cdots \dfrac{W_{d_n}^{i_n}(z_n)}{2^{d_n}\sqrt{2}}.
\end{equation*}
Since we know how to integrate every $W_d^i(z)$,
we simply define
\begin{equation}
\label{eq:Fgn_anc}
F_{g,n}(z_1,\dots,z_n):=\sum_{\vec{d},\;\vec{i}}\cor{\bar{\tau}_{d_1}(\tilde{e}_{i_1})\dots\bar{\tau}_{d_n}(\tilde{e}_{i_n})}_g \dfrac{\theta_{d_1}^{i_1}(z_1)}{2^{d_n}\sqrt{2}}\cdots \dfrac{\theta_{d_n}^{i_n}(z_n)}{2^{d_n}\sqrt{2}}.
\end{equation}
Then
from Proposition \ref{prop:theta},
 we see that \eqref{eq:primitive}
 and \eqref{eq:skew} are automatically 
 satisfied. 
We also know from
 Proposition \ref{prop:theta} 
 that  the $x^{-1}$-expansion of $F_{g,n}$ near $z_1=\dots=z_n=0$ is given by
\begin{multline*}
F_{g,n}(x_1,\dots,x_n)\\ 
=\sum_{\vec{d},\;\vec{i}}\cor{\bar{\tau}_{d_1}(\tilde{e}_{i_1})\dots\bar{\tau}_{d_n}(\tilde{e}_{i_n})}_g
\prod_{k=1}^n A^{i_k}_{\mu_k}\sum_{m=d}^{\infty}(\cS_{m-d})^{\mu_k}_{\nu_k}\ \left( - \delta^{\nu_k}_{1}\delta^{m}_{0}\, \dfrac{1}{2} - \delta^{\nu_k}_{2}\,  m!\, \dfrac{1}{x_k^{m+1}} \right).
\end{multline*}
Using \eqref{eq:anc_desc_rel} and \eqref{eq:vtrel}, we find
\begin{equation}
\begin{aligned}
\label{eq:Fgn_desc}
&F_{g,n}(x_1,\dots,x_n)\\ 
&=\sum_{\vec{d},\vec{i}}\cor{\tau_{d_1}(e_{\mu_1})\dots\tau_{d_n}(e_{\mu_n})}_g \prod_{i=1}^n \left( - \delta^{\mu_i}_{1}\delta^{d_i}_{0}\, \dfrac{1}{2} - \delta^{\mu_i}_{2}\,  d_i!\, \dfrac{1}{x^{d_i+1}} \right)\\ 
&=
\la \prod_{i=1}^n \left(-\frac{\tau_0(1)}{2}-\sum_{b=0}^\infty\frac{b! \tau_{b} (\omega) }{x_i^{b+1}} \right) \ra_{g,n}.
\end{aligned}
\end{equation}
The final condition \eqref{eq:initial} follows from
the fact that $\cor{(\tau_0(1))^n}_{g,n}=0$
for all $g$ and $n$ in the stable range. 
This concludes the proof of the theorem.


\section{The shift of variable simplification}
\label{sec:shift-of-variable}

Let us now turn our attention toward
proving \eqref{eq:QCE} of
Theorem~\ref{thm:main}. In this 
section, as the first step, we establish
a formula for  the wave function
$\Psi(x,\hbar)$ of \eqref{eq:Psi} 
involving only  the stationary Gromov-Witten
invariants. 

Our starting point is 
\begin{multline}
\label{eq:psi-as-pert-sum}
 \log \Psi(x,\hbar) = \frac{1}{\hbar} S_0(x) + S_1(x) 
\\
 \phantom{ =\ } 
+ \sum_{g,d=0}^\infty \sum_{\substack{n=1\\ 2g-2+n>0}}^\infty \frac{\hbar^{2g-2+n}}{n!} 
\la \left(-\frac{\tau_0(1)}{2}-\sum_{b=0}^\infty\frac{b! \tau_{b} (\omega) }{x^{b+1}} \right)^n \ra^d_{g,n} .
\end{multline}
Using the string equation and some earlier results in~\cite{DMSS12},
we shall give an expression
for $\log\Psi(x,\hbar)$ purely in terms of the stationary sector. 
More precisely, we prove the following lemma.

\begin{lemma} 
\label{lem:shift-of-variable}
The function 
$\log \Psi(x,\hbar)$ is a solution to 
 the following
difference 
equation:
\begin{multline}
\label{eq:shift-simplification-precise}
 \exp\left(-\frac{\hbar}{2}\frac{d}{dx}\right) \log \Psi(x,\hbar) = \frac{1}{\hbar}\left(x-x\log x\right)
\\
 \phantom{ =\ } 
+ \sum_{g,d=0}^\infty \sum_{n=1}^\infty \frac{\hbar^{2g-2+n}}{n!} 
\la \left(-\sum_{b=0}^\infty\frac{b! \tau_{b} (\omega) }{x^{b+1}} \right)^n \ra^d_{g,n} .
\end{multline}
\end{lemma}

\subsection{Expansion of $S_0$ and $S_1$}

The functions $S_0(x)$ and $S_1(x)$ 
of \eqref{eq:S0-intro} 
and \eqref{eq:S1-intro} 
are derived from the first steps of the WKB method, that is, they are just imposed by the quantum spectral curve equation. 
In this subsection,
 we represent them in terms of the 
unstable $(0,1)$- and $(0,2)$-Gromov-Witten invariants. 

First let us calculate these functions from 
the WKB approximation \eqref{eq:WKB}.
After taking the semi-classical limit
\eqref{eq:SCL}, 
we can calculate $S_1'(x)$ as follows:
\begin{align*}
&e^{-\frac{1}{\hbar}S_0(x)-S_1(x)}
(e^{\hbar\frac{d}{dx}}+e^{-\hbar\frac{d}{dx}}-x)
 e^{\frac{1}{\hbar}S_0(x)+S_1(x)}
 \\
&=e^{S_0'(x)+\hbar\left(\frac{1}{2}
S_0''(x)+S_1'(x)\right)}
e^{\hbar\frac{d}{dx}}+
e^{-S_0'(x)+\hbar
\left(\frac{1}{2}S_0''(x)-S_1'(x)\right)}
e^{-\hbar\frac{d}{dx}}-x+O(\hbar^2)\\
&=e^{S_0'(x)}\left(1+\hbar
\left(\frac{1}{2}S_0''(x)+S_1'(x)\right)\right)
+e^{-S_0'(x)}\left(1+\hbar
\left(\frac{1}{2}S_0''(x)-S_1'(x)\right)\right)
\\
&\qquad -x 
+O(\hbar^2)\\
&=\hbar\left(\frac{S_0''(x)}{2}
\left(e^{S_0'(x)}+e^{-S_0'(x)}\right)+S_1'(x)\left(e^{S_0'(x)}-e^{-S_0'(x)}\right)
\right)+O(\hbar^2).
\end{align*}
The coefficient of $\hbar$ must vanish, hence we can solve for $S_1'(x)$.  Since $$S_0''(x)=\frac{d}{dx}S_0'(x)=\frac{d}{dx}\log z=\frac{\frac{d}{dz}\log z}{x'(z)}	=\frac{\frac{1}{z}}{1-\frac{1}{z^2}}=\frac{1}{z-\frac{1}{z}},$$
we find
\begin{equation}
\label{eq:S1'}
S_1'(x)=-\frac{1}{2}\frac{1}{z-\frac{1}{z}}\frac{z+\frac{1}{z}}{z-\frac{1}{z}}=-\frac{1}{2}\frac{z(z^2+1)}{(z^2-1)^2}.
\end{equation}
It is proved in~\cite[Equation~(7.9) and Theorem~7.7]{DMSS12} that
\begin{equation}
\begin{aligned}
\label{eq:DMSS-S0}
 \sum_{d=0}^\infty 
\la  \left(-\sum_{b=0}^\infty\frac{b! \tau_{b} (\omega) }{x^{b+1}} \right) \ra^d_{0,1} 
&=
\sum_{d=1}^\infty 
\la  \left(-\frac{(2d-2)! \tau_{2d-2} (\omega) }{x^{2d-1}} \right) \ra^d_{0,1} 
\\  
 &= -2z+\left(z+\frac{1}{z}\right)\log \left(1+z^2\right),
 \end{aligned}
 \end{equation}
and 
\begin{equation}
\label{eq:DMSS-S1}
\sum_{d=0}^\infty 
\la 
\prod_{i=1}^2
 \left(-\sum_{b=0}^\infty\frac{b! \tau_{b} (\omega) }{x_i^{b+1}} \right) \ra^d_{0,2} 
=
-\log\left(1-z_1z_2\right).
\end{equation}
One of the implications of the string equation is 
\begin{equation*}
\la \tau_0(1)\tau_{b+1}(\omega) \ra^{d}_{0,2} = \la \tau_{b}(\omega) \ra^{d}_{0,1}.
\end{equation*}
Using this form of the string equation 
 and Equation~\eqref{eq:DMSS-S0}, we  calculate
  that
  \begin{equation}
\begin{aligned} 
\label{eq:cross}
& \sum_{d=1}^\infty 
\la \left(-\frac{1}{2}\tau_0(1)\right) \left(-\frac{(2d-1)! \tau_{2d-1} (\omega) }{x_i^{2d}} \right) \ra^d_{0,2} 
\\ 
& = \frac{1}{2} \frac{d}{dx} \left(-2z+\left(z+\frac{1}{z}\right)\log \left(1+z^2\right)\right) 
\\ 
& = \frac{1}{2}\log x+\frac{1}{2}\log z.
\end{aligned}
\end{equation}
Note that the only condition we have for 
$S_0(x)$ is that $S'_0(x) = \log z$. Therefore,
if we  define 
\begin{equation*}
S_0(z) := F_{0,1}(z) =\int\W_{0,1}(z)
=\int y(z)dx(z)
\end{equation*}
by formally applying \eqref{eq:Sm-intro}
for $m=0$, and impose 
the skew-symmetry condition \eqref{eq:skew}
to the primitive $F_{0,1}(z)$,
then from \eqref{eq:DMSS-S0} we obtain
\begin{equation}
\begin{aligned}\label{eq:S0-x}
S_0(x)& = \frac{1}{z}-z +\left(z+\frac{1}{z}\right) \log z
\\ 
& = \left(x-x\log x\right) + 
\sum_{d=1}^\infty 
\la  \left(-\frac{(2d-2)! \tau_{2d-2} (\omega) }{x_i^{2d-1}} \right) \ra^d_{0,1} .
\end{aligned}
\end{equation}

The determination of $S_1(x)$ is trickier. 
Morally speaking, if we formally apply 
\eqref{eq:Sm-intro} for $m=1$, then
we obtain
\begin{equation}
\label{eq:S1-z}
S_1(x) = -\half F_{0,2}(z,z)
\end{equation}
for the primitive
\begin{multline}
\label{eq:02primitive}
F_{0,2}(z_1,z_2)
= \int^{z_1}\int^{z_2}
\W_{0,2}(z_1,z_2)
\\
=
 \int^{z_1}\int^{z_2} \left(
 \frac{dz_1dz_2}{(z_1-z_2)^2}
-\frac{dx_1dx_2}{(x_1-x_2)^2}\right)
\\
 =-\log(1-z_1z_2) + f(z_1) + f(z_2) +c.
\end{multline}
Here we are imposing the condition that
$F_{0,2}(z_1,z_2)$ is a symmetric function.
The fact that $F_{0,2}$ is a primitive of
$\W_{0,2}$ does not determine
 the function $f(z)$. Therefore, we are
free to choose $f(z)$ so that the differential
equation \eqref{eq:S1'} holds. Obviously,
we need to choose 
$f(z) = \half \log z$. 
In this way, using \eqref{eq:DMSS-S1}
and \eqref{eq:cross} as well, we obtain
\begin{equation}
\begin{aligned}\label{eq:S1-x}
S_1(x)& = -\frac{1}{2}\log\left(1-z^2\right) +\frac{1}{2}\log z
\\ 
& = -\frac{1}{2}\log x + \frac{1}{2}\sum_{d=0}^\infty 
\la \left(-\frac{\tau_0(1)}{2}-\sum_{b=0}^\infty\frac{b! \tau_{b} (\omega) }{x^{b+1}} \right)^2 \ra^d_{0,2}.
\end{aligned}
\end{equation}

\begin{remark}
This adjustment of the choice of $S_1(x)$
also appears in the Hitchin fibration case
of \cite{DM13}. Still we have one degree of 
freedom for choosing a constant $c$ of
\eqref{eq:02primitive}. It does not matter
to the linear 
quantum curve equation \eqref{eq:QCE},
because the constant term $c$ only affects 
on the overall constant factor of $\Psi$
of \eqref{eq:Psi}.
\end{remark}

\subsection{A new formula for $\log\Psi$} We use Equations~\eqref{eq:S0-x} and~\eqref{eq:S1-x} 
 to rewrite  the formula \eqref{eq:psi-as-pert-sum}
  for $\log\Psi$   in the following way: 
\begin{equation}
\label{eq:newPsi}
 \log\Psi(x,\hbar)=\sum_{g,d=0}^\infty \sum_{n=1}^\infty \frac{\hbar^{2g-2+n}(-1)^n}{n!} \Theta^d_{g,n},
\end{equation}
where
\begin{multline}
\label{eq:Theta-010-def}
\Theta_{0,1}^0  := -x+x\log x + \frac{\hbar}{2}\log x 
\\
+
\sum_{k=2}^\infty \la \tau_0(1)^k\tau_{k-2}(\omega)\ra_{0,k+1}^0 \frac{(-1)^k\hbar^{k}}{2^k k!}\frac{(k-2)!}{x^{k-1}}
\end{multline}
and
\begin{equation}
\label{eq:Theta-gnd-def}
\Theta_{g,n}^d  :=  \sum_{k=0}^\infty \sum_{b_1,\dots,b_n=0}^\infty 
\la \tau_0(1)^k \prod_{i=1}^n \tau_{b_i}(\omega) \ra^d_{g,n+k} \frac{(-1)^k\hbar^{k}}{2^k k!} \frac{\prod_{i=1}^n b_i!}{x^{n+\sum_{i=1}^n b_i}}
.
\end{equation}
It is obvious that for dimensional reasons,
 $\Theta_{0,n}^0=0$ for any $n\geq 2$. 
Lemma~\ref{lem:shift-of-variable} is then a direct corollary to the following statement.

\begin{lemma} \label{lem:Theta} 
The quantities defined in \eqref{eq:Theta-010-def}
and \eqref{eq:Theta-gnd-def} are given by
\begin{align}\label{eq:Theta-010}
\Theta_{0,1}^0 & = -\left(x+\frac{\hbar}{2}\right)+\left(x+\frac{\hbar}{2}\right)\log \left(x+\frac{\hbar}{2}\right); \\ \label{eq:Theta-gnd}
\Theta_{g,n}^d & = \sum_{b_1,\dots,b_n} 
\la \prod_{i=1}^n \tau_{b_i}(\omega) \ra^d_{g,n} \frac{\prod_{i=1}^n b_i!}{\left(x+\frac{\hbar}{2}\right)^{n+\sum_{i=1}^n b_i}},
\end{align}
where in the second equation the sum is taken over all $b_1,\dots,b_n\geq 0$ such that $\sum_{i=1}^n b_i = 2g+2d -2$.
\end{lemma}

\subsection{Proof of Lemma~\ref{lem:Theta}} Since
the difference between the definitions
\eqref{eq:Theta-010-def}-\eqref{eq:Theta-gnd-def}
and the values 
\eqref{eq:Theta-010}-\eqref{eq:Theta-gnd}
is simply the elimination of $\tau_0(1)$, 
we prove
 Lemma~\ref{lem:Theta} by using the  string equation for the Gromov-Witten invariants of $\P$:
\begin{equation}\label{eq:string-red}
\la \tau_0(1)^{k} \prod_{i=1}^n \tau_{b_i}(\omega) \ra^d_{g,n+k}
= \sum_{\substack{j=1 \\ b_j > 0}}^n 
\la \tau_0(1)^{k-1} \tau_{b_j-1}(\omega) \prod_{\substack{i=1 \\ i\not = j}}^n \tau_{b_i}(\omega) \ra^d_{g,n+k-1},
\end{equation}
where we assume $2g-2+n>1$ and $k>0$. 

First, let us directly compute  $\Theta_{0,1}^0$. Equation~\eqref{eq:string-red} implies that 
\begin{equation}
\la \tau_0(1)^k\tau_{k-2}(\omega)\ra_{0,k+1}^0=\la \tau_0(1)^{k-1}\tau_{k-3}(\omega)\ra_{0,k}^0=
\la \tau_0(1)^2 \tau_0(\omega)\ra_{0,3} =1.
\end{equation}
Therefore,
\begin{multline*}
 \sum_{k=2}^\infty \la \tau_0(1)^k\tau_{k-2}(\omega)\ra_{0,k+1}^0 \frac{(-1)^k\hbar^{k}}{2^k k!}\frac{(k-2)!}{x^{k-1}}
\\ \notag
 =
\sum_{k=2}^\infty \frac{(-1)^k\hbar^{k}}{2^k k!}\frac{(k-2)!}{x^{k-1}}
= \left(x+\frac{\hbar}{2}\right) \log \left( \frac{x+\frac{\hbar}{2}}{x} \right) -\frac{\hbar}{2}.
\end{multline*}
This proves Equation~\eqref{eq:Theta-010}.

The proof of Equation~\eqref{eq:Theta-gnd}
goes as follows. Recall that $g+d>0$ and $n>0$.  Equation~\eqref{eq:string-red} implies that any correlator $\la \tau_0(1)^{k} \prod_{i=1}^n \tau_{b_i}(\omega) \ra^d_{g,n+k}$ can be represented as a linear combination of the correlators $\la \prod_{i=1}^n \tau_{b_i}(\omega) \ra^d_{g,n}$ with $\sum_{i=1}^n b_i = 2g+2d-2$. Moreover, for any $k\geq 0$ and $c_1,\dots,c_n\geq  0$ such that $\sum_{i=1}^n c_i = k$, the coefficient of a particular correlator $\la \prod_{i=1}^n \tau_{b_i}(\omega) \ra^d_{g,n}$ in $\la \tau_0(1)^k \prod_{i=1}^n \tau_{b_i+c_i}(\omega) \ra^d_{g,n+k}$
is equal to 
\begin{equation*}
\frac{k!}{c_1!\cdots c_n!} .
\end{equation*}
Therefore, the total coefficient of 
$\la \prod_{i=1}^n \tau_{b_i}(\omega) \ra^d_{g,n} $ in $\Theta^d_{g,n}$ is equal to
\begin{equation}
\begin{aligned} \label{eq:cf-1}
& \sum_{k=0}^\infty \sum_{\substack{c_1,\dots,c_n=0\\ c_1+\cdots c_n=k}}^\infty 
\frac{(-1)^k\hbar^{k}}{2^k k!} \frac{\prod_{i=1}^n (b_i+c_i)!}{x^{n+\sum_{i=1}^n (b_i+c_i)}}
\frac{k!}{c_1!\cdots c_n!} 
\\ 
& =\frac{\prod_{i=1}^n (b_i)!}{x^{n+\sum_{i=1}^n (b_i)}} \sum_{k=0}^\infty \left( \frac{-\hbar}{2x} \right)^k
\sum_{\substack{c_1,\dots,c_n\ge 0\\ c_1+\cdots +c_n=k}}\prod_{i=1}^n\frac{(b_i+c_i)!}{b_i!c_i!}.
\end{aligned}
\end{equation}
On the other hand, expansion of the coefficient of $\la \prod_{i=1}^n \tau_{b_i}(\omega) \ra^d_{g,n} $ in Equation~\eqref{eq:Theta-gnd} is equal to 
\begin{equation}
\begin{aligned}\label{eq:cf-2}
& \frac{\prod_{i=1}^n (b_i)!}{\left(x+\frac{\hbar}{2}\right)^{n+\sum_{i=1}^n (b_i)}}
= \prod_{i=1}^n \frac{(b_i)!}{\left(x+\frac{\hbar}{2}\right)^{b_i+1}}
\\ 
& = \prod_{i=1}^n \frac{(b_i)!}{\left(x\right)^{b_i+1}} \sum_{c_i=0}^\infty \left(\frac{-\hbar}{2x}\right)^k
\frac{(b_i+c_i)!}{b_i!c_i!}.
\end{aligned}
\end{equation}
Since \eqref{eq:cf-1} and~\eqref{eq:cf-2} 
are identical,
we have proved  Equation~\eqref{eq:Theta-gnd}.
This completes the proof of 
Lemma~\ref{lem:shift-of-variable}.




\section{Reduction to the semi-infinite wedge formalism}
\label{sec:reduction-to-semi-infinite-wedge}

In this Section we represent the formula for 
$\Psi(x,\hbar)$ in terms of the se\-mi-in\-fi\-nite wedge formalism. We use the formula of Okounkov-Pand\-ha\-ri\-pan\-de \cite{OP06}
that relates the stationary sector of the Gromov-Witten invariants
 of $\P$ to the expectation values of the so-called $\E$-operators. In order to include the extra combinatorial factors that we have in the expansion of $\log\Psi(x,\hbar)$, we consider the $\E$-operators with values in formal differential operators.

\subsection{Semi-infinite wedge formalism} 
In this subsection we recall very briefly some basic facts about the semi-infinite wedge formalism. For more details we refer to~\cite{DKOSS13,OP06,SSZ12}.

Let us consider a vector space $V:=\bigoplus_{c=-\infty}^\infty V_c$, where $V_c$ is spanned by the basis vectors 
$\underline{a_1}\wedge \underline{a_2}\wedge \underline{a_3} \wedge \cdots$ such that $a_i\in \Z+1/2$, $i=1,2,\dots$, $a_1>a_2>a_3\dots$, and for all but a finite number of terms we have $a_i=1/2-i+c$.
We denote by $\psi_k$ the operator $\underline{k}\wedge\colon V_c\to V_{c+1}$, and by $\psi_k^*$ the operator $\partial/\partial\underline{k}\colon V_c\to V_{c-1}$. Both are odd operators, and they satisfy the graded commutaion relation $[\psi_i,\psi_i^*]=1$, with all other possible pairwise commutators equal to zero.

We denote by $:\psi_i\psi^*_j:$ the normally ordered product, that is, $:\psi_i\psi^*_j:=\psi_i\psi^*_j$ for $j>0$ and  $:\psi_i\psi^*_j:=-\psi^*_j\psi_i$ for $j<0$. We introduce the operators $\E_n(z)$, $n\in\Z$ as
\begin{equation}
\E_n(z):=\sum_{k\in \Z+1/2} \exp\left(z\left(k-\frac{n}{2}\right)\right) :\psi_{k-r}\psi^*_k: +\frac{\delta_{n0}}{\zeta(z)},
\end{equation}
where $\zeta(z)=\exp(z/2)-\exp(-z/2)$. These operators satisfy the commutation relation 
$[\E_n(z),\E_m(w)]=\zeta(nw-mz)\E_{n+m}(z+w)$.

For any operator $\mathcal{A}=\E_{n_1}(z_1)\cdots\E_{n_m}(z_m)$ we denote by $\lv\mathcal{A}\rv$ the coefficient of the vector $v_\emptyset:=-\underline{1/2}\wedge-\underline{3/2}\wedge-\underline{5/2}\wedge\cdots$ in the basis expansion of $\mathcal{A}v_\emptyset$. If we want to compute a particular correlator $\lv \E_{n_1}(z_1)\cdots\E_{n_m}(z_m)\rv$, 
 first we use the commutation relation for the 
$\E$-operators, and then
appeal to the  simple fact that $\left.\left. \E_n(z)\rv=0$ for $n>0$, $\lv \E_n(z) \right.\right.=0$ for $n<0$, and $\lv \E_0(z_1)\cdots \E_0(z_n)\rv = 1/\left(\zeta(z_1)\cdots \zeta(z_n)\right)$.
In this section
 we are mostly interested in  correlators for the form
 \begin{equation}
 \label{eq:A}
 \lv \cA\rv= 
 \lv\E_1(0)^d \prod_{i=1}^n \E_0(z_i) \E_{-1}(0)^d\rv.
\end{equation}

For the purpose of establishing the results
in \cite{OP06}, Okounkov and 
Pandharipande considered the \emph{disconnected}
version of 
Gromov-Witten invariants and Hurwitz numbers.
The disconnectedness here means we allow
disconnected domain curves mapped to $\P$.
For example, they establish in
\cite[Proposition 3.1, Equation 3.4]{OP06} 
a formula for disconnected stationary 
Gromov-Witten invariants of $\P$, which reads
\begin{equation}
\label{eq:OP3.1}
\sum_{b_1,\dots,b_n\ge -2} 
\la \prod_{i=1}^n \tau_{b_i}(\omega)\ra
 ^{\bullet \;d}
\prod_{i=1}^n x_i ^{b_i+1}
=
\frac{1}{(d!)^2}
\lv\E_1(0)^d \prod_{i=1}^n \E_0(x_i) \E_{-1}(0)^d\rv,
\end{equation}
where $\la \;\;\;\ra^\bullet$ denotes the
disconnected Gromov-Witten invariant.
Counting the number of disconnected domain
curves and connected ones are related simply
by talking the logarithm. Thus we have
\begin{multline*}
\sum_{g=0}^\infty
\sum_{b_1,\dots,b_n\ge -2} 
\la \prod_{i=1}^n \tau_{b_i}(\omega)\ra_{g,n}
 ^{d}
\prod_{i=1}^n x_i ^{b_i+1}
\\
=
\log\left(
\sum_{b_1,\dots,b_n\ge -2} 
\la \prod_{i=1}^n \tau_{b_i}(\omega)\ra
 ^{\bullet \;d}
\prod_{i=1}^n x_i ^{b_i+1}\right).
\end{multline*}
This prompts us to introduce the
\emph{connected} correlator notation, 
corresponding to \eqref{eq:OP3.1}, as follows:
\begin{multline}
\label{eq:connected}
\sum_{g=0}^\infty
\sum_{b_1,\dots,b_n\ge -2} 
\la \prod_{i=1}^n \tau_{b_i}(\omega)\ra_{g,n}
 ^{d}
\prod_{i=1}^n x_i ^{b_i+1}
\\
=
\frac{1}{(d!)^2}
\lv\E_1(0)^d \prod_{i=1}^n \E_0(x_i) 
\E_{-1}(0)^d\rv^\circ.
\end{multline}
The connected correlator is also known 
as the \emph{cumulant} in probability theory,
which is calculate via the inclusion-exclusion 
formula. In general, for an operator $\cA$
of  \eqref{eq:A}, 
we denote by $\lv \mathcal{A} \rv^{\circ}$ the contribution coming from the 
single operator of the form
 $\E_0(\sum_{i=1}^nz_i)$ in the end, after applying
 the commutation relation successively.
 Of course in terms of generating functions,
 this simply means we take the logarithm of 
 the expression.
 See \cite[Definition 2.12, Definition 2.14]{DKOSS13} for more detail.

\subsection{A new formula for $\Psi$}

Noticing that
 $\exp\left(\frac{\hbar}{2}\frac{d}{dx}\right)$
is an automorphism, 
from 
\eqref{eq:shift-simplification-precise}
we find 
\begin{equation*}
\log \Psi(x,\hbar)=\exp\left(\frac{\hbar}{2}\frac{d}{dx}\right) T(x),
\end{equation*}
where
\begin{multline*}
T(x) :=\sum_{g,d=0}^\infty \sum_{n=1}^\infty \frac{\hbar^{2g-2+n}}{n!} \sum_{b_1,\dots,b_n=0}^\infty
\la \prod_{i=1}^n \tau_{b_i}(\omega) \ra^d_{g,n} \prod_{i=1}^n \left(-\frac{b_i!}{x^{b_i+1}}\right)
\\  
+ \frac{1}{\hbar}\la \tau_{-2}(\omega) \ra^0_{0,1} \left(x-x\log x\right).
\end{multline*}
Here we have 
used the convention of~\cite{OP06} that $\la\tau_{-2}(\omega)\ra^0_{0,1}=1$ and 
$\tau_{-1}(\omega) = 0$.
We are now ready to re-write the right-hand side
in terms of expectation values of $\E$-operators.
Corollary~\ref{cor:operator} of the 
 following lemma is the main result of this section.
\begin{lemma}\label{lem:operator} For any $d\geq 0$, $n\geq 1$, $(d,n)\not=(0,1)$, we have
\begin{multline}
\label{eq:dn}
\sum_{g=0}^\infty \hbar^{2g-2+n} \sum_{b_1,\dots,b_n=0}^\infty
\la \prod_{i=1}^n \tau_{b_i}(\omega) \ra^d_{g,n} \prod_{i=1}^n \left(-\frac{b_i!}{x_i^{b_i+1}}\right)
\\ 
 =\frac{1}{(d!)^2\hbar^{2d}} \lv \E_1(0)^d \prod_{i=1}^n \E_0\left(-\hbar \frac{\partial}{\partial x_i}\right)\left(\log x_i\right)\E_{-1}(0)^{-d} \rv^\circ.
\end{multline}
For  $d=0$ and $n=1$, we have
\begin{multline}\label{eq:d0n1}
 \frac{1}{\hbar}\la \tau_{-2}(\omega) \ra^0_{0,1} \left(x-x\log x\right)
+ \sum_{g=1}^\infty \hbar^{2g-1} 
\la \prod_{i=1}^n \tau_{2g-2}(\omega) \ra^0_{g,1} \left(-\frac{(2g-2)!}{x^{2g-1}}\right)
\\ 
 =\lv \E_0\left(-\hbar \frac{d}{dx}\right)\left(\log x\right)\rv^\circ.
\end{multline}
\end{lemma}

Here  we denote by $\la\;\;\;\ra^\circ$ the connected expectation value. This means that
after the successive application of the
commutation relation, all differential operators
appear in one correlator.
 Of course for $d=0$, $n=1$, we have
  $\la \E_0\ra^\circ=\la \E_0\ra$.
The following corollary is a straightforward
application of Lemma~\ref{lem:operator}.

\begin{corollary}\label{cor:operator} We have the following expression for $\log\Psi$:
\begin{multline}\label{eq:logpsi}
\log\Psi(x,\hbar)
\\=
 \sum_{d=0}^\infty \frac{1}{\hbar^{2d} (d!)^2}\lv \E_1(0)^d \sum_{n=1}^\infty \frac{\left(\exp\left(\frac{1}{2}\hbar \frac{d}{dx}\right)\E_0\left(-\hbar \frac{d}{dx}\right) \left(\log x\right)\right)^n}{n!} \E_{-1}(0)^d \rv^\circ.
\end{multline}
\end{corollary}

\subsection{Proof of Lemma~\ref{lem:operator}}

The starting point of the proof is \eqref{eq:connected}.
Note that only negative $b_i$ contribution
comes from 
$\la \tau_{-2}(\omega) \ra^{0}_{0,1}=1$, which is  the coefficient of $x_i^{-1}$ in $\lv \E_0\left(x_i\right)\rv^\circ$.

Let $A(x)=\sum_{i=-1}^{\infty} a_i x^i$ be an arbitrary Laurent series. Observe that
\begin{equation}
\label{eq:observation}
A\left(-\hbar \frac{d}{dx}\right) \left(\log x \right) = a_{-1}\left(\frac{x-x\log x}{\hbar}\right) + a_0\log x - \sum_{i=1}^{\infty} a_i \frac{(i-1)!\hbar^{i}}{x^{i}}.
\end{equation}
We can apply this observation to the correlator 
\begin{equation}
\frac{1}{(d!)^2}\lv \E_1(0)^d \prod_{i=1}^n \E_0\left(x_i\right)\E_{-1}(0)^{-d} \rv ^\circ
\end{equation}
and change $\E_0\left(x_i\right)$ to 
$$
\E_0\left(-\hbar \frac{\partial}{\partial x_i}\right)\log x_i.$$
If $(n,d)\not=(1,0)$, then
we have a formal Laurent series in $x_1,\dots,x_n$, where the degree of each variable in each term is 
less than or equal to 
$-1$. Together with the computation of the degree of $\hbar$, which is $\sum_{i=1}^n (b_i +1) -2d = 2g-2+n$, we establish Equation~\eqref{eq:dn}.

If $(n,d)=(1,0)$, then it is sufficient
 to observe that $\lv \E_0\left(x\right)\rv^\circ=x^{-1} +O(x)$. Thus we have one additional term $(x-x\log x)/\hbar$ as in \eqref{eq:observation}, which is exactly the first term in Equation~\eqref{eq:d0n1}.

This completes the proof of 
Lemma~\ref{lem:operator}, and hence, Corollary~\ref{cor:operator}.




\section{Reduction to a combinatorial problem}
\label{sec:combinatorial}

The expression \eqref{eq:logpsi} of 
$\log \Psi$ in the form of 
the vacuum expectation value of the 
operator product allows us to 
convert the quantum curve equation
\eqref{eq:QCE} into a combinatorial formula.

Our starting point 
is the $\Psi$-function represented in the form 
\begin{equation}
\label{eq:Psi in operator}
\Psi(x,\hbar)=1+\sum_{d=0}^\infty \frac{1}{\hbar^{2d}(d!)^2}\lv \E_1(0)^d \sum_{n=1}^\infty \frac{1}{n!}\A(x)^n \E_{-1}(0)^d \rv^\star,
\end{equation}
where
\begin{equation}
\begin{aligned}
\label{eq:A-operator}
\A(x)& =\exp\left(\frac{\hbar}{2} \frac{d}{dx}\right) \E_0\left(-\hbar \frac{d}{dx}\right) \left(\log x\right) \\ 
& = \sum_{k\in\Z+\frac{1}{2}} \exp\left(\left(-k+\frac{1}{2}\right)\hbar\frac{d}{dx}\right)\left(\log x\right) :\psi_k\psi_k^*: 
\\ 
& \phantom{ =\ }
+ B\left(-\hbar \frac{d}{dx}\right) \left(\frac{x-x\log x}{\hbar}\right).
\end{aligned}
\end{equation}
Here $B(t):= t/(e^t-1)$ in
\eqref{eq:A-operator} is the generating series of the Bernoulli numbers, and the notation $\la - \ra^\star$ 
in \eqref{eq:Psi in operator}
means that in the computation of this expectation value 
using the commutation relations, 
we never allow any $\E_1(0)$ and  $\E_{-1}(0)$ to 
commute directly. We need this requirement since we exponentiate the series~\eqref{eq:logpsi},
which does not
have terms without $\E_0$-operators.
The goal of this section is to prove 
Corollary~\ref{cor:EqX}.

\begin{lemma} \label{lem:Xd} We have
\begin{equation} \label{eq:Xd}
\exp\left(\frac{1}{\hbar^2}\right)\Psi(x,\hbar)= \exp\left(B\left(-\hbar \frac{d}{dx}\right) \left(\frac{x-x\log x}{\hbar}\right)\right) X,
\end{equation}
where $X:=\sum_{d=0}^\infty X_d/\hbar^{2g}$, and $X_d$ is given by
\begin{equation}
\begin{aligned}\label{eq:Comb}
X_d&=\frac{1}{(d!)^2}
\lv \E_1(0)^d \exp\left( \sum_{k\in\Z+\frac{1}{2}} \log \left(x-\left(k-\frac{1}{2}\right)\hbar\right) :\psi_k\psi_k^*:\right) \E_{-1}(0)^d\rv
\\ 
&=\sum_{\lambda\vdash d} 
\left(
\frac{\dim\lambda}{d!}\right)^2 
\prod_{i=1}^\infty \frac{x+(i-\lambda_i)\hbar}{x+i\hbar}
.
\end{aligned}
\end{equation}
\end{lemma}

\begin{corollary}\label{cor:EqX} The quantum spectral curve equation
\begin{equation*}
\left[\exp\left(\hbar\frac{d}{dx}\right)+\exp\left(-\hbar\frac{d}{dx}\right)-x\right] \Psi(x,\hbar) = 0
\end{equation*}
is equivalent to the following equation for the function $X$:
\begin{equation}
\label{eq:EqX1}
\left[\frac{1}{x+\hbar} \exp\left({\hbar\frac{d}{dx}}\right)+x \exp\left(-{\hbar\frac{d}{dx}}\right)-x\right]X=0.
\end{equation}
\end{corollary}

\begin{proof}[Proof of Lemma~\ref{lem:Xd}] Corollary~\ref{cor:operator} implies that 
\begin{align}
\label{eq:log-psi-infiniwedge}
&\sum_{d=0}^\infty \frac{
\lv \E_1(0)^d 
\exp\left( 
\exp\left(\frac{1}{2}\hbar \frac{d}{dx}\right)\E_0\left(-\hbar \frac{d}{dx}\right) \left(\log x\right)\right)
\E_{-1}(0)^d 
\rv^\circ
}{\hbar^{2d} (d!)^2} 
\\ \notag
& = \log \Psi(x,\hbar) + \frac{1}{\hbar^2} +1.
\end{align}
Indeed, we add terms with $n=0$, and it is easy to see that  
$$\lv \E_1(0)^d \E_{-1}(0)^d \rv^\circ =0,
\qquad d\ge 2,
$$ 
and $\lv \E_1(0) \E_{-1}(0)\rv^\circ=\lv \mathrm{Id} \rv^\circ=1$. Therefore,
\begin{align}
&\exp\left(\frac{1}{\hbar^2}\right)\Psi(x,\hbar) =\\ \notag
& \sum_{d=0}^\infty \frac{
\lv \E_1(0)^d 
\exp\left( 
\exp\left(\frac{1}{2}\hbar \frac{d}{dx}\right)\E_0\left(-\hbar \frac{d}{dx}\right) \left(\log x\right)\right)
\E_{-1}(0)^d 
\rv
}{\hbar^{2d} (d!)^2} .
\end{align}

From the definition of the operator $\E_0$, 
we have
\begin{equation}
\begin{aligned}
\label{eq:E0def}
& \exp\left(\frac{1}{2}\hbar \frac{d}{dx}\right)\E_0\left(-\hbar \frac{d}{dx}\right) \left(\log x\right)
\\ 
& = \exp\left(\frac{1}{2}\hbar \frac{d}{dx}\right) 
\left(
\sum_{k\in \Z+1/2} \log\left(x-k\hbar\right) :\psi_{k}\psi^*_k:  
\right)
\\ 
& \phantom{ =\ }
+\exp\left(\frac{1}{2}\hbar \frac{d}{dx}\right) 
\frac{-\hbar\frac{d}{dx}}{\exp\left(-\frac{1}{2}\hbar \frac{d}{dx}\right)-\exp\left(\frac{1}{2}\hbar \frac{d}{dx}\right)}
\left(\frac{x-x\log x}{\hbar}\right)
\\ 
& = 
\sum_{k\in \Z+1/2} \log\left(x-\left(k-\frac{1}{2}\right)\hbar\right) :\psi_{k}\psi^*_k:  
+
B\left(-\hbar \frac{d}{dx}\right)\left(\frac{x-x\log x}{\hbar}\right).
\end{aligned}
\end{equation}
Now define
\begin{align}
\label{eq:A1}
A_1&=\sum_{k\in \Z+1/2} \log\left(x-\left(k-\frac{1}{2}\right)\hbar\right) :\psi_{k}\psi^*_k: \\
\label{eq:A2}
A_2 &=B\left(-\hbar \frac{d}{dx}\right)\left(\frac{x-x\log x}{\hbar}\right).
\end{align}
Since $A_1$ and $A_2$ commute, we have $\exp(A_1+A_2)=\exp(A_2)\exp(A_1)$. Furthermore, since $A_2$ is a scalar operator, we have
\begin{multline*}
 \sum_{d=0}^\infty \frac{
\lv \E_1(0)^d 
\exp(A_2)\exp(A_1)
\E_{-1}(0)^d 
\rv
}{\hbar^{2d} (d!)^2} 
\\ 
 = \exp(A_2) \sum_{d=0}^\infty \frac{
\lv \E_1(0)^d 
\exp(A_1)
\E_{-1}(0)^d 
\rv
}{\hbar^{2d} (d!)^2}.
\end{multline*}
This is exactly the right-hand side of Equation~\eqref{eq:Xd}.
\end{proof}

\begin{proof}[Proof of Corollary~\ref{cor:EqX}] We just have to show that
\begin{align*}
\exp(-A_2) \exp\left({\hbar\frac{d}{dx}}\right) \exp(A_2)
& = \frac{1}{x+\hbar} \exp\left({\hbar\frac{d}{dx}}\right) ; \\
\exp(-A_2) \exp\left(
{-\hbar\frac{d}{dx}}
\right)
 \exp(A_2)
& = x \exp\left({-\hbar\frac{d}{dx}}\right) ; \\
\exp(-A_2) x \exp(A_2)
& = x.
\end{align*}
The last equality is tautological, and the first two are obtained by a straightforward computation.
\end{proof}

For completeness, let us also explain Equation~\eqref{eq:Comb}.  It is based on several standard facts about the semi-infinite wedge formalism. For any partition $\lambda=(\lambda_1\geq \lambda_2\geq \lambda_3\geq \dots)$ we associate a basis vector $v_\lambda\in V_0$ given by
\begin{equation}
\underline{\left(\lambda_1-\frac{1}{2}\right)} \wedge \underline{\left(\lambda_2-\frac{3}{2}\right)}
\wedge \underline{\left(\lambda_3-\frac{5}{2}\right)}
\wedge\cdots .
\end{equation}
Then, we have  $\E_{-1}(0)^d v_\emptyset = \sum_{\lambda\vdash d} \dim \lambda \cdot v_\lambda$, 
$\lv \E_{1}(0)^d v_\lambda \right.\right.=\dim\lambda$, and the fact that for any constants $a_n$, $n\in\Z+1/2$,  $v_\lambda$ is an eigenvector of the operator $\sum_{n\in\Z+1/2} a_n :\psi_n\psi_n^*:$ with the eigenvalue $\sum_{i=1}^\infty \left(a_{\lambda_i-i+1/2}-a_{-i+1/2}\right)$. Therefore,  $v_\lambda$ is an 
eigenvector of the operator 
\begin{equation}
A_1=\exp\left( \sum_{k\in\Z+\frac{1}{2}} \log \left(x-\left(k-\frac{1}{2}\right)\hbar\right) :\psi_k\psi_k^*:\right) 
\end{equation}
with the eigenvalue
\begin{equation}
\exp\left(
\sum_{i=1}^\infty \log\left(x+(i-\lambda_i)\hbar\right)-\log\left(x+i\hbar\right)
\right)=\prod_{i=1}^\infty \frac{x+(i-\lambda_i)\hbar}{x+i\hbar},
\end{equation}
and the total weight of the vector $v_\lambda$ in $\lv \E_1(0)^d A_1 \E_{-1}(0)^d\rv$ is $\left(\dim \lambda \right)^2$. This implies Equation~\eqref{eq:Comb}.




\section{Key combinatorial argument}
\label{sec:key}

We have shown that the quantum curve equation
\eqref{eq:QCE} is equivalent to a 
combinatorial equation \eqref{eq:EqX1}, 
which is indeed a first-order recursion
equation for $X_d$ of 
\eqref{eq:Comb} with respect to the index $d$.
In this section we prove \eqref{eq:EqX1}.

Let $\lambda\vdash d$ be a partition $\lambda=(\lambda_1\geq \lambda_2\geq ...\geq \lambda_{\ell(\lambda)}>0)$ of $d\geq 1$. We can always append it with $d-\ell(\lambda)$ zeros $\lambda_{\ell(\lambda)+1}:=0,\dots,\lambda_d:=0$
at the end so that we would have a partition of $d$ of  length $d$ with non-negative parts. Throughout this section we  use this convention that a partition of $d$ has length $d$. 

Consider the following sum over all partitions $\lambda=(\lambda_1\geq \lambda_2\geq ...\geq \lambda_d)$ of $d\geq 1$:
\begin{equation}
X_d:=\sum_{\lambda\vdash d} \frac{1}{H_\lambda^2} \prod_{i=1}^d
\frac{x+(i-\lambda_i)\hbar}{x+i\hbar} .
\end{equation}
Here $H_\lambda:=\prod_{ij} h_{ij}$, where $h_{ij}$ is the hook length at the vertex $(ij)$ of the corresponding Young diagram, so that $d!/\prod h_{ij}$ is the dimension of the irreducible representation corresponding to $\lambda$. Or equivalently, it is the number of the standard Young tableaux of this shape. We use the convention 
that $X_0:=1$.

In this Section we prove the following key combinatorial lemma.

\begin{lemma}\label{lm:key-lemma} The series $X:=\sum_{d=0}^\infty X_d/\hbar^{2g}$ satisfies the following equation:
\begin{equation}\label{eq:key-lemma}
\left[\frac{1}{x+\hbar} \exp\left({\hbar\frac{d}{dx}}\right)+x \exp\left(-{\hbar\frac{d}{dx}}\right)-x\right]X=0.
\end{equation}
\end{lemma}

\begin{proof}
In fact, \eqref{eq:key-lemma} is a direct 
consequence of the following more refined 
statement.

\begin{lemma}\label{lm:key-lemma-d} For any $d\geq 1$ we have
\begin{equation}\label{eq:key-lemma-d}
\frac{1}{x/\hbar+1} \exp\left({\hbar\frac{d}{dx}}\right)X_{d-1}+\left[\frac{x}{\hbar} \exp\left(-{\hbar\frac{d}{dx}}\right)-\frac{x}{\hbar}\right]X_d=0.
\end{equation}
\end{lemma}

Indeed, since $\left[x \exp\left(-{h\frac{d}{dx}}\right)-x\right]X_0=0$, the sum  of Equation~\eqref{eq:key-lemma-d} for all $d\geq 1$ with coefficients $1/\hbar^{2d-1}$ yields Lemma~\ref{lm:key-lemma}.
\end{proof}

To prove Lemma~\ref{lm:key-lemma-d}, we need to recall some standard facts on the hook length formula as well as a recent result of Han~\cite{H10}.

\subsection{Hook lengths and shifted parts of partition}

We use the following result from~\cite{H10}. For a partition $\lambda\vdash d$, $d\geq1$, we define the so-called $g$-function:
\begin{equation}
g_\lambda(y):=\prod_{i=1}^d (y+\lambda_i-i).
\end{equation}
For any $\lambda\vdash d$, $d\geq 1$, we denote by $\lambda\setminus 1$ the set of all partitions of $d-1$ that can be obtained from $\lambda$ (or rather the corresponding Young diagram) by removing one corner of $\lambda$. 

\begin{lemma}[Han~\cite{H10}] For every partition $\lambda$ we have
\begin{equation}
\frac{1}{H_\lambda}\left(g_\lambda(y+1)-g_\lambda(y)\right) 
= \sum_{\mu\in\lambda\setminus 1} \frac{1}{H_\mu} g_\mu(y).
\end{equation}
Here $y$ is a formal variable.
\end{lemma}

We need the following corollary of this lemma.

\begin{corollary}\label{cor:H-square}
For an integer $d\geq 1$ we have
\begin{equation}
\sum_{\lambda\vdash d+1}\frac{1}{H_\lambda^2}\left(g_\lambda(y+1)-g_\lambda(y)\right) 
= \sum_{\mu\vdash d} \frac{1}{H_\mu^2} g_\mu(y).
\end{equation}
\end{corollary}

\begin{proof}
We recall that for any $\mu\vdash d$, $d\geq 1$, we have:  
\begin{equation}
\sum_{\substack{\lambda\vdash d+1\\ \lambda\setminus 1 \ni \mu}} \frac{1}{H_\lambda} = \frac{1}{H_\mu}.
\end{equation}
Therefore, 
\begin{align*}
\sum_{\mu\vdash d} \frac{1}{H_\mu^2} g_\mu(y) 
& = \sum_{\mu\vdash d} \frac{1}{H_\mu}  \sum_{\substack{\lambda\vdash d+1\\ \lambda\setminus 1 \ni \mu}} \frac{1}{H_\lambda}  g_\mu(y) \\ 
& = \sum_{\lambda\vdash d+1} \frac{1}{H_\lambda}  \sum_{\substack{\mu\vdash d\\ \mu\in\lambda\setminus 1}} \frac{1}{H_\lambda}  g_\mu(y) \\ 
&=\sum_{\lambda\vdash d+1} \frac{1}{H_\lambda^2}  \left(g_\lambda(y+1)-g_\lambda(y)\right).
\end{align*}
\end{proof}

\subsection{Reformulation of Lemma~\ref{lm:key-lemma-d} in terms of $g$-functions} 

We make the following substitution: $y:=-x/\hbar$. Then we see that
\begin{equation*}
X_d=\sum_{\lambda\vdash d} \frac{1}{H_\lambda^2} 
\frac{g_\lambda(y)}{\prod_{i=1}^d (y-i)}.
\end{equation*}
Moreover,
\begin{equation}
\begin{aligned}\label{eq:x-y}
& \frac{1}{x/\hbar+1} \exp\left({\hbar\frac{d}{dx}}\right)X_{d-1}+\left[\frac{x}{\hbar} \exp\left(-{\hbar\frac{d}{dx}}\right)-\frac{x}{\hbar}\right]X_d
\\  &
= \frac{-1}{y-1} \exp\left(-\frac{d}{dy}\right)X_{d-1}+\left[-y\exp\left(\frac{d}{dy}\right)+y\right]X_d.
\end{aligned}
\end{equation}
Observe that
\begin{align} \label{eq:d-1}
 \frac{-1}{y-1} \exp\left(-\frac{d}{dy}\right)X_{d-1} 
 &= -\sum_{\lambda\vdash d-1} \frac{1}{H_\lambda^2} 
\frac{g_\lambda(y-1)}{\prod_{i=1}^d (y-i)}; 
\\
\notag
 -y\exp\left(\frac{d}{dy}\right) X_d &= (d-y)\sum_{\lambda\vdash d} \frac{1}{H_\lambda^2} 
\frac{g_\lambda(y+1)}{\prod_{i=1}^d (y-i)}; 
\\
\notag
 y X_d &= y \sum_{\lambda\vdash d} \frac{1}{H_\lambda^2} 
\frac{g_\lambda(y)}{\prod_{i=1}^d (y-i)}.
\end{align}
Using Corollary~\ref{cor:H-square} we can rewrite the right hand side of Equation~\eqref{eq:d-1} as 
\begin{equation}
\frac{-1}{y-1} \exp\left(-\frac{d}{dy}\right)X_{d-1} = \sum_{\lambda\vdash d} \frac{1}{H_\lambda^2} 
\frac{g_\lambda(y-1) - g_\lambda(y)}{\prod_{i=1}^d (y-i)}.
\end{equation}
Therefore, the right hand side of Equation~\eqref{eq:x-y} is equal to
\begin{equation}
\frac{Y_d(y)}{\prod_{i=1}^d (y-i)},
\end{equation}
where
\begin{equation}
Y_d(y):=\sum_{\lambda\vdash d} \frac{(d-y)g_\lambda(y+1)+(y-1)g_\lambda(y)+g_\lambda(y-1)}{H_\lambda^2}.
\end{equation}
Note that $Y_d(y)$ is a polynomial in $y$ of degree $\leq d+1$, and Lemma~\ref{lm:key-lemma-d} is equivalent to the following statement:

\begin{lemma}\label{lm:key-lemma-Y} For any $d\geq 1$ we have $Y_d(y)\equiv 0$.
\end{lemma}

\subsection{Proof of Lemma~\ref{lm:key-lemma-Y}} In this subsection we prove Lemma~\ref{lm:key-lemma-Y} and, therefore, Lemmas~\ref{lm:key-lemma-d} and~\ref{lm:key-lemma}.

First of all, it is easy to check that for any $d\geq 1$ the polynomial $Y_d(y)$ has at least one root. Namely, 
\begin{equation}
\label{eq:Ydd}
Y_d(d)=\sum_{\lambda\vdash d} \frac{(d-1)g_\lambda(d)+g_\lambda(d-1)}{H_\lambda^2} = 0.
\end{equation}
Indeed, $g_\lambda(d)$ is not equal to zero only for $\lambda=(1,1,\dots,1)$. In this case $g_\lambda(d)=d!$, $H_\lambda= d!$, and $(d-1)g_\lambda(d)/H_\lambda^2 = (d-1)/d!$.
Notice that $g_\lambda(d-1)$ does not vanish only for $\lambda=(2,1,1,\dots,1,0)$. In this case
$g_\lambda(d-1)=-d\cdot (d-2)!$, $H_\lambda=d\cdot (d-2)!$, and $g_\lambda(d-1)/H_\lambda^2=-(d-1)/d!$. Thus we see that  always $Y_d(d)=0$,
establishing \eqref{eq:Ydd}.

Now we proceed by induction. It is easy to check that $Y_1(y)\equiv 0$.
Assume that we know that $Y_d(y)\equiv 0$. Corollory~\ref{cor:H-square} then implies that 
\begin{align*}
 Y_d(y) 
&=\sum_{\lambda\vdash d} \frac{(d-y)g_\lambda(y+1)+(y-1)g_\lambda(y)+g_\lambda(y-1)}{H_\lambda^2}
\\ 
&= \sum_{\lambda\vdash d+1} \frac{(d-y)g_\lambda(y+2)+(2y-d-1)g_\lambda(y+1)+ (2-y)g_\lambda(y)-g_\lambda(y-1)}{H_\lambda^2}
\\ 
& = \sum_{\lambda\vdash d+1} \frac{((d+1)-(y+1))g_\lambda(y+2)+((y+1)-1)g_\lambda(y+1)+g_\lambda(y)}{H_\lambda^2}
\\ 
& \phantom{ =\  } 
- \sum_{\lambda\vdash d+1} \frac{((d+1)-y)g_\lambda(y+1)+ (y-1)g_\lambda(y)+g_\lambda(y-1)}{H_\lambda^2}
\\ 
& = Y_{d+1}(y+1)-Y_{d+1}(y).
\end{align*}
By assumption, we have $Y_d(y)\equiv 0$.
Therefore, 
 $Y_{d+1}(y+1)=Y_{d+1}(y)$ for any $y$. Hence$Y_{d+1}$ is constant. Since we have shown that $Y_{d+1}(d+1)=0$, we conclude that $Y_{d+1}\equiv 0$.

This completes the 
proof of Lemmas~\ref{lm:key-lemma-Y},~\ref{lm:key-lemma-d}, and~\ref{lm:key-lemma}.
Thus we have established the main theorem of 
this paper.

\section{Laguerre polynomials}
\label{sec:Laguerre}

In this section we prove a  combinatorial expression for the functions
\begin{equation*}
X_d:=\sum_{\lambda\vdash d} \frac{1}{H_\lambda^2} \prod_{i=1}^d
\frac{x+(i-\lambda_i)\hbar}{x+i\hbar}
\end{equation*}
in terms of the \emph{Laguerre polynomials} $L_n^{(\alpha)}(z)$.

The Laguerre
 polynomial is a solution of the 
 differential equation
\begin{equation*}
z\frac{d^2}{dz^2}L_n^{(\alpha)}(z)+(\alpha+1-z)\frac{d}{dz}L_n^{(\alpha)}(z)+nL_n^{(\alpha)}(z)=0,
\end{equation*}
and has a closed expression 
\begin{equation*}
L_n^{(\alpha)}(z)=\sum_{i=0}^n (-1)^i \binom{n+a}{n-i} \frac{z^i}{i!}.
\end{equation*}
Here is a list of some properties of the Laguerre polynomials:
\begin{align} \label{eq:Eric}
\frac{L_n^{\alpha}(z)}{\binom{n+\alpha}{n}} &
= 1-\sum_{j=1}^n \frac{z^j}{a+j} \frac{L_{n-j}^{(j)}(z)}{(j-1)!}
; \\ \label{eq:nna}
nL_n^{\alpha}(z) &
= (n+\alpha)L_{n-1}^{\alpha}(z) - z L_{n-1}^{\alpha+1}(z)
; \\ \label{eq:diff}
L_n^{\alpha}(z) &
= L_{n}^{\alpha+1}(z) - L_{n-1}^{\alpha+1}(z)
; \\ \label{eq:binom}
\binom{n+\alpha}{n} &
= \sum_{i=0}^n \frac{z^i}{i!}L_{n-i}^{(\alpha+i)}(z).
\end{align}

\begin{proposition} \label{prop:Laguerre}For any $d\geq 0$ we have:
\begin{equation*}
X_d=\frac{1}{d!}\left(1-\sum_{m=1}^d \frac{1}{(m-1)!} L_{d-m}^{(m)}(1) \frac{\hbar}{x+m\hbar}\right).
\end{equation*}
\end{proposition}

\begin{remark} This equation can be rewritten as 
\begin{equation*}
X_d=\frac{1}{d!} \frac{L_d^{x/\hbar}(1)}{L_d^{x/\hbar}(0)}.
\end{equation*}
Indeed, we  just apply the identity~\eqref{eq:Eric} for $\alpha=x/\hbar$ and $z=1$ and further observe that $\binom{n+x/\hbar}{n}=L_n^{(x/\hbar)}(0)$.
\end{remark}

\begin{proof}[Proof of Proposition~\ref{prop:Laguerre}] It is obvious that both $X_d$ and 
\begin{equation}\label{eq:laguerre}
\tilde X_d:=\frac{1}{d!}\left(1-\sum_{m=1}^d \frac{1}{(m-1)!} L_{d-m}^{(m)}(1) \frac{\hbar}{x+m\hbar}\right)
\end{equation}
are rational functions with simple poles at $x/\hbar = -1,-2,\dots,-d$. We have defined $X_0:=1$, and it is easy to see that $Z_0=1$. Then we know (see Lemma~\ref{lm:key-lemma-d}) that all $X_d$ are unambiguously determined by the equation
\begin{equation}\label{eq:key-lemma-d-1}
\frac{1}{\frac{x}{\hbar}+1} \exp\left({\hbar\frac{d}{dx}}\right)X_{d}+\left[\frac{x}{\hbar} \exp\left(-{\hbar\frac{d}{dx}}\right)-\frac{x}{\hbar}\right]X_{d+1}=0
\end{equation}
for all $d\geq 0$.
In order to prove the proposition we check that $\{\tilde X_d\}_{d\geq 0}$ also satisfy this equation.

Indeed, observe that 
\begin{align}\label{eq:term1}
& \frac{1}{\frac{x}{\hbar}+1} \exp\left({\hbar\frac{d}{dx}}\right)\tilde X_{d}
\\ \notag
& =  \frac{1}{d!} \left(\frac{1}{\frac{x}{\hbar} +1} -\sum_{m=1}^d \frac{L_{d-m}^{(m)}(1)}{(m-1)!}\frac{1}{(\frac{x}{\hbar}+1)(\frac{x}{\hbar}+m+1)}  \right)
\\ \notag
& = \frac{1}{\frac{x}{\hbar} +1} \frac{1}{d!} \left(1- \sum_{m=1}^d \frac{L_{d-m}^{(m)}(1)}{m!}\right) 
+ \frac{1}{d!}\sum_{m=2}^{d+1} \frac{L_{d+1-m}^{(m-1)}(1)}{(m-1)!}\frac{1}{(\frac{x}{\hbar}+m)} ;
\end{align}
\begin{align}\label{eq:term2}
& \frac{x}{\hbar} \exp\left(-{\hbar\frac{d}{dx}}\right) \tilde X_{d+1} 
\\ \notag
& = \frac{1}{(d+1)!} \left(\frac{x}{\hbar} -\sum_{m=1}^{d+1} \frac{L_{d+1-m}^{(m)}(1)}{(m-1)!}\frac{\frac{x}{\hbar}}{\frac{x}{\hbar}+m-1}  \right)
\\ \notag
& = \frac{\frac{x}{\hbar}}{(d+1)!} -\frac{1}{(d+1)!}\sum_{m=1}^{d+1} \frac{L_{d+1-m}^{(m)}(1)}{(m-1)!}
+\frac{1}{(d+1)!}\sum_{m=1}^{d} \frac{L_{d-m}^{(m+1)}(1)}{(m-1)!}\frac{1}{\frac{x}{\hbar}+m} ;
\end{align}
and
\begin{align}\label{eq:term3}
& -\frac{x}{\hbar} \tilde X_{d+1} 
= \frac{1}{(d+1)!} \left(-\frac{x}{\hbar} +\sum_{m=1}^{d+1} \frac{L_{d+1-m}^{(m)}(1)}{(m-1)!}\frac{\frac{x}{\hbar}}{\frac{x}{\hbar}+m}  \right)
\\ \notag
& = \frac{-\frac{x}{\hbar}}{(d+1)!} +\frac{1}{(d+1)!}\sum_{m=1}^{d+1} \frac{L_{d+1-m}^{(m)}(1)}{(m-1)!}
-\frac{1}{(d+1)!}\sum_{m=1}^{d+1} \frac{L_{d+1-m}^{(m)}(1)}{(m-1)!}\frac{m}{\frac{x}{\hbar}+m} .
\end{align}

It is obvious that in the sum of the expressions~\eqref{eq:term1}, \eqref{eq:term2}, and~\eqref{eq:term3} the coefficient of $x/\hbar$ and the constant term vanish. So, we have to prove that the coefficient of each $1/(x/\hbar+m)$, $m=1,\dots,d+1$, vanishes. 

The coefficient of $1/(x/\hbar+d+1)$ is equal to
\begin{equation*}
\frac{1}{d!}\frac{L_{0}^{(d)}(1)}{d!}-\frac{1}{(d+1)!}\frac{(d+1)L_{0}^{(d+1)}(1)}{d!} ,
\end{equation*}
which is equal to zero since $L_0^{(d)}=L_0^{(d+1)}=1$.

The coefficient of $1/(x/\hbar+m)$, $2\leq m\leq d$, is equal to
\begin{equation}\label{eq:coeff-m}
\frac{1}{d!}\frac{L_{d+1-m}^{(m-1)}(1)}{(m-1)!}+\frac{1}{(d+1)!}\frac{L_{d-m}^{(m+1)}(1)}{(m-1)!}
-\frac{1}{(d+1)!} \frac{mL_{d+1-m}^{(m)}(1)}{(m-1)!}.
\end{equation}
First we use Equation~\eqref{eq:nna} for $z=1$, $\alpha=m$, and $n=d+m-1$:
\begin{equation}
(d+1-m)L_{d+1-m}^{(m)}(1)=(d+1)L_{d-m}^{(m)}(1) - L_{d-m}^{(m+1)}(1).
\end{equation}
We see then that expression~\eqref{eq:coeff-m} is equal to
\begin{equation*}
\frac{1}{d!(m-1)!}\left( L_{d+1-m}^{(m-1)}(1)-L_{d+1-m}^{(m)}(1)+L_{d-m}^{(m)}(1)\right).
\end{equation*}
And this is equal to zero due to Equation~\eqref{eq:diff} for $n=d+1-m$, $\alpha = m-1$, and $z=1$.

The coefficient of $1/(x/\hbar+1)$ is equal to
\begin{equation*}
\frac{1}{d!} \left(1- \sum_{m=1}^d \frac{L_{d-m}^{(m)}(1)}{m!}\right)+
\frac{1}{(d+1)!}{L_{d-1}^{(2)}(1)}
-\frac{1}{(d+1)!}L_{d}^{(1)}(1).
\end{equation*}
Using that $L_{d-1}^{(2)}(1)=-dL_d^{(1)}(1)+(d+1)L_{d-1}^{(1)}(1)$ (which is Equation~\eqref{eq:nna} for $z=1$, $n=d$, and $\alpha=1$)
and $-L_d^{(1)}(1)+L_{d-1}^{(1)}(1)=-L_{d}^{(0)}(1)$ (which is Equation~\eqref{eq:diff} for $z=1$, $d=n$, and $\alpha=0$), we see that this coefficient is equal to
\begin{equation*}
\frac{1}{d!} \left(1- \sum_{m=0}^d \frac{L_{d-m}^{(m)}(1)}{m!}\right).
\end{equation*}
This is equal to zero due to Equation~\eqref{eq:binom} for $z=1$, $\alpha=0$, and $n=d$.

Thus we see that the sum of expressions~\eqref{eq:term1}, \eqref{eq:term2}, and~\eqref{eq:term3} is equal to zero. So, the functions $\tilde X_d$ satisfy Equation~\eqref{eq:key-lemma-d-1}, and, therefore, $\tilde X_d=X_d$ for all $d\geq 0$.
\end{proof}

\section{Toda lattice equation}
\label{sec:Toda}

In this section we recall, for completeness, the Toda lattice equation for the partition function of the Gromov-Witten invariants of $\P$. We show that the Toda lattice equation implies a quadratic relation for the functions $X_d$, $d\geq 0$, considered in the previous sections. It is an open question whether it is possible to relate the Toda lattice equation to the quantum spectral curve equation.

It is convenient to include a degree variable $q$ in the free energy of the Gromov-Witten invariants of $\P$.  Define
\begin{equation}
\mathcal F_g :=  \sum_dq^d\left\langle\exp\left\{\sum_{i\geq 0}^{\infty}\tau_i(\omega)t_i+\tau_0(1)t\right\}\right\rangle_g^d
\end{equation}
(where we have switched off $\tau_k(1)$ for $k>0$).
Then $\mathcal F=\sum_{g=0}^\infty \mathcal F_g$ satisfies the Toda lattice equation:
\begin{equation}
\label{eq:toda-lattice}
\exp(\mathcal F(t+1)+\mathcal F(t-1)-2\mathcal F(t))=\frac{1}{q}\frac{\partial^2}{\partial t_0^2}\mathcal F(t),
\end{equation}
which was conjectured by Eguchi-Yang \cite{EY94} and proven by Dubrovin-Zhang \cite{DZ04} and Okounkov-Pandharipande \cite[Equation~(4.11)]{OP06}.

We specialise
\begin{equation}
\Phi(x,\hbar,q):=F\left(q=\frac{q}{\hbar^2},\ t=-\frac{1}{2},\ t_i=-i!\left(\frac{\hbar}{x}\right)^{i+1}\right)
\end{equation} 
and consider the function $\exp\left\{\Phi(x,\hbar,q)-\Phi(x,\hbar,0)\right\}$.

\begin{lemma} We have:
\begin{equation}\label{eq:X-series}
\exp\left\{\Phi(x,\hbar,q)-\Phi(x,\hbar,0)\right\}=\sum_{d=0}^\infty \left(\frac{q}{\hbar^2}\right)^dX_d.
\end{equation}
\end{lemma}

\begin{proof}

Indeed, tracing back the arguments of Sections~\ref{sec:shift-of-variable} and~\ref{sec:reduction-to-semi-infinite-wedge} and taking $q$ into account this time, it is easy to see that
\begin{align}
\Phi(x, \hbar, q)=\sum_{d=0}^\infty \lb \frac{q}{\hbar^2} \rb^d \sum_{g=0}^\infty \la \exp \lb -\frac{1}{2} \tau_0(1) - \sum_{i=0}^\infty \tau_i(\omega) i! \lb\frac{\hbar}{x}\rb^{i+1} \rb \ra^{d}_g.
\end{align}
The proof of Lemma~\ref{lem:Xd} implies that
\begin{align}
\label{eq:phi-zero-b}
\Phi(x, \hbar, 0) & = \sum_{g=0}^\infty \la \exp \lb -\frac{1}{2} \tau_0(1) - \sum_{i=0}^\infty \tau_i(\omega) i! \lb\frac{\hbar}{x}\rb^{i+1} \rb \ra^{ 0}_g \\ \notag
& = B\lb-\hbar \frac{d}{dx} \rb \lb \frac{x - x \log x}{\hbar} \rb
\end{align}
These two observations and Equation~\eqref{eq:Xd} imply Equation~\eqref{eq:X-series}.
\end{proof}

By abuse of notation, we denote the function $\exp\left\{\Phi(x,\hbar,q)-\Phi(x,\hbar,0)\right\}$ also by $X$ (so-called degree-weighted $X$). The quantum spectral curve equation for this degree-weighted $X$ reads
\begin{equation}
\left[\frac{q}{x+\hbar} \exp\left({\hbar\frac{d}{dx}}\right)+x \exp\left(-{\hbar\frac{d}{dx}}\right)-x\right]X=0.
\end{equation}

The Toda lattice equation combined with the string and the divisor equations implies the following equation for $X$:

\begin{proposition} We have:
\begin{equation}  \label{eq:toda}
\frac{X(x+\hbar)X(x-\hbar)}{X(x)^2}=\frac{x+\hbar}{x}\frac{\partial}{\partial q}\left(q\frac{\partial}{\partial q}\log X(x)\right).
\end{equation}
\end{proposition}

\begin{proof} We recall the divisor equation 
\begin{equation}
\label{eq:divisor}
\frac{\partial \mathcal F}{\partial t_0}=\frac{1}{2}t^2+q\frac{\partial \mathcal F}{\partial q}.
\end{equation}
Consider Equations~\eqref{eq:toda-lattice}. The result of Section \ref{sec:shift-of-variable} implies that the shifts in $\tau_0(1)$-coefficient $t$ can be replaced by shifts of variable $x$ with factor $\hbar$. Using this and Equation~\eqref{eq:divisor} we obtain equation for
$\Phi(x, \hbar, q)$:
\begin{equation}
\exp\lb\Phi(x + \hbar,\hbar,q) + \Phi(x - \hbar,\hbar,q) - 2 \Phi(x,\hbar,q)\rb
= \frac{\partial}{\partial q} \lb q \frac{\partial}{\partial q}  \Phi(x,\hbar,q) \rb.
\end{equation}

From the definition of degree-weighted $X$ it follows that 
\begin{equation}
\frac{\partial}{\partial q}\left(q\frac{\partial}{\partial q}\log X(x)\right) = \frac{\partial}{\partial q} \lb q \frac{\partial}{\partial q}  \Phi_q(x) \rb.
\end{equation}
Therefore, in order to prove Equation~\eqref{eq:toda}, it is enough to show that
\begin{equation}
\label{eq:toda-remnant}
\frac{\exp\left({\Phi(x + \hbar,\hbar,0)}\right) \exp\left({\Phi(x - \hbar,\hbar,0)}\right)}
{\exp\left({2 \Phi(x,\hbar,0)}\right)} = \frac{x}{x + \hbar}.
\end{equation}
Indeed, using Equation~\eqref{eq:phi-zero-b} we can represent the left hand side of Equation~\eqref{eq:toda-remnant} in the following form:
\begin{align}
& \exp\lb \left. \left[ \lb e^{t} + e^{-t} - 2 \rb B(t) \right] \right|_{t = -\hbar \frac{d}{dx} } \lb \frac{x - x\log x}{\hbar} \rb\rb
\\ \notag & = \exp\lb \left. \left[ \lb e^{t/2} - e^{-t/2} \rb^2 \frac{t}{e^t - 1} \right] 
\right |_{t = -\hbar \frac{d}{dx}} 
\lb \frac{x - x\log x}{\hbar} \rb\rb \\ \notag & = \exp \lb \log (x) - \log (x + \hbar) \rb.
\end{align}
\end{proof}

The homogeneous in $q$ part of the Toda lattice equation~\eqref{eq:toda} for the series expansion given by Equation~\eqref{eq:X-series} can be rewritten as 
\begin{equation}
\label{eq:post-toda-quadratic}
\frac{x}{x + \hbar} \sum_{a + b = d} X_a(x + \hbar) X_b(x -\hbar) = \sum_{a + b = d + 1} X_a(x) X_b(x) \frac{(a - b)^2}{2}
\end{equation}
for any $d\geq 0$. 
We can use Equation~\eqref{eq:key-lemma-d}  in order to rewrite this equation as 
\begin{align}
\label{eq:post-toda-quadratic-one-level}
& \sum_{a + b = d + 1} \lb X_a(x) X_b(x - \hbar) - X_a(x - \hbar) X_b(x - \hbar) \rb 
\\ \notag
& = \frac{\hbar^2}{x^2} \sum_{a + b = d + 1} X_a(x) X_b(x) \frac{(a - b)^2}{2},
\end{align}
for any $d\geq 0$. It would be interesting to see whether Equation~\eqref{eq:post-toda-quadratic} or Equation~\eqref{eq:post-toda-quadratic-one-level} can be related to Equation~\eqref{eq:key-lemma-d}.

\end{document}